\documentclass[12pt]{article}
\usepackage[margin=1in]{geometry}

\usepackage{amsmath,amsfonts}
\usepackage{algorithmic}
\usepackage{algorithm}
\usepackage{array}
\usepackage{textcomp}
\usepackage{stfloats}
\usepackage{url}
\usepackage{verbatim}
\usepackage{graphicx}
\usepackage{amsthm}  
\newtheorem{theorem}{Theorem}
\usepackage{comment}
\usepackage{booktabs}
\usepackage{subcaption}
\usepackage{cite}
\usepackage{amssymb}
\usepackage{comment} 
\usepackage{nomencl}
\usepackage{float}
\usepackage{multirow}
\usepackage{soul}
\usepackage{xcolor}
\sethlcolor{yellow}
\usepackage{etoolbox}
\renewcommand\nomgroup[1]{%
  \item[\bfseries
  \ifstrequal{#1}{A}{Parameters}{%
  \ifstrequal{#1}{B}{Functions}{%
  \ifstrequal{#1}{C}{Decision Variables}{%
   }}}%
  ]}
\makenomenclature
\begin{document}
\title{Efficient Resource Allocation under Adversary Attacks: A Decomposition-Based Approach}
\author{Mansoor~Davoodi~and~Setareh~Maghsudi
\thanks{Faculty of Electrical Engineering and Information Technology,
  Ruhr-University Bochum, 44801 Bochum, Germany. 
  \{Mansoor.DavoodiMonfared, Setareh.Maghsudi\}@ruhr-uni-bochum.de}
}
\maketitle
\begin{abstract}
We address the problem of allocating limited resources in a network under persistent yet statistically unknown adversarial attacks. Each node in the network may be degraded, but not fully disabled, depending on its available defensive resources. The objective is twofold: to minimize total system damage and to reduce cumulative resource allocation and transfer costs over time.
We model this challenge as a bi-objective optimization problem and propose a decomposition-based solution that integrates chance-constrained programming with network flow optimization. The framework separates the problem into two interrelated subproblems: determining optimal node-level allocations across time slots, and computing efficient inter-node resource transfers. We theoretically prove the convergence of our method to the optimal solution that would be obtained with full statistical knowledge of the adversary.
Extensive simulations demonstrate that our method efficiently learns the adversarial patterns and achieves substantial gains in minimizing both damage and operational costs, comparing three benchmark strategies under various parameter settings.
\end{abstract}

\textbf{Keywords:}
Resource allocation, Adversary, Decomposition, Bi-objective optimization, Chance-constrained optimization, Network flow.

\section{Introduction}

Modern distributed systems—spanning cloud infrastructures, communication networks, and cyber-physical systems—must remain operational under persistent threats. One critical challenge is the efficient allocation of limited defensive or operational resources across a network, especially when the system is exposed to adversarial attacks that degrade node functionality without full observability or predictability. We focus on the problem of Resource Allocation in the Presence of Adversaries (RAPA), where an adversary targets nodes in a network with unknown but stationary probabilistic behavior.

The fundamental difficulty in RAPA lies in striking a balance between two conflicting objectives: (\textit{i}) minimizing the total system damage caused by adversarial attacks, and (\textit{ii}) minimizing the operational costs of resource allocation and transfer. These goals must be met under uncertainty about the adversary's strategy—specifically, the statistical distribution of attack intensities across nodes is unknown, including both mean and variance. This makes the problem substantially more complex than deterministic or fully informed resource allocation problems commonly studied in the literature.

Existing studies in resource allocation under adversarial settings often make strong assumptions such as full observability of attack statistics or adversary models that are either fully deterministic or follow known stochastic distributions. In contrast, our work targets a more realistic and challenging setting where only the outcomes of attacks are observable over time, and their underlying statistical properties must be inferred indirectly during operation.

\textbf{Contributions.}
This paper introduces a novel bi-objective optimization framework for RAPA, characterized by the following contributions:

\begin{itemize}
    \item We formulate RAPA as a bi-objective stochastic optimization problem aimed at minimizing cumulative damage and resource allocation/transfer costs across multiple time slots. 
    \item To address this problem efficiently, we develop an iterative decomposition approach that splits the problem into two coupled subproblems:
    \begin{enumerate}
        \item The first subproblem determines node-level resource allocations. It is reformulated as a chance-constrained optimization problem to handle uncertainty about attack distributions.
        \item The second subproblem focuses on minimizing inter-node resource transfer costs using network flow optimization.
    \end{enumerate}
    This decomposition enables scalable and tractable solutions while maintaining interaction between the subproblems, which we justify both algorithmically and theoretically.
    \item Unlike most prior work, our method does not require prior knowledge of attack means or variances. Instead, it learns adversarial behavior iteratively while still making effective allocation decisions at each time step.
    \item We prove that our method converges to the optimal solution that would have been obtained if the attack distributions were fully known.
    \item We conduct extensive simulation experiments comparing our approach to three baselines:
    \begin{itemize}
        \item A statistical benchmark with known mean and variance of attacks,
        \item A greedy heuristic based on node importance and attack frequency,
        \item An oracle with full knowledge of the attack distribution for all nodes and time slots.
    \end{itemize}
    Our method significantly outperforms greedy approach and approaches the benchmark with known statistics, demonstrating its ability to learn and adapt effectively in uncertain adversarial settings.
\end{itemize}

To the best of our knowledge, this is the first study to address the adversarial resource allocation problem using a joint linear-quadratic optimization framework under fully unknown attack statistics. By integrating chance-constrained modeling, adaptive decision-making, and network optimization, we contribute a rigorous yet practical solution to a class of problems critical for secure and resilient systems. The framework is broadly applicable and scalable, opening pathways for future research in robust distributed control, cyber-defense, and intelligent network management under uncertainty.

\section{Related Work}

Efficient resource allocation under adversarial conditions has been studied across various domains, including cybersecurity, wireless communication, cloud computing, and edge systems. Broadly, existing methods fall into three primary categories: game-theoretic frameworks, learning-based approaches, and optimization-based techniques. Each category offers distinct advantages and limitations depending on the availability of information about the adversary and the decision-maker’s objectives.

{Game-Theoretic Approaches.} 
Game-theoretic models have been widely applied to adversarial decision-making problems, especially when strategic interactions between a defender and an attacker are explicit. Representative works include \cite{yang2013improving, chang2021resource, wu2022dynamic}, which model resource allocation as Stackelberg or Bayesian games. These approaches typically assume rational adversaries and known payoff structures, which enable equilibrium-based strategies. However, these assumptions are often too strong in dynamic or data-sparse environments. Game-theoretic models also generally require complete or partial knowledge of the adversary’s utility functions or action space, which is not available in the problem setting we address.

{Learning-Based Approaches.}
Learning-based solutions, particularly those based on online learning and reinforcement learning, aim to adapt to adversarial environments by using historical data. For instance, \cite{du2022adversarial} and \cite{alqerm2020enhanced} utilize deep learning and adaptive heuristics to infer attack patterns and adjust resource allocations accordingly. Cai et al. \cite{cai2023deep} apply deep reinforcement learning to dynamic resource environments. While these methods are powerful when large-scale data is available, they may lack performance guarantees and often require extensive training, which is costly in high-risk or real-time systems. Furthermore, many learning-based methods do not explicitly handle resource transfer costs or offer convergence guarantees.

{Optimization-Based Approaches.}
Optimization techniques offer a more principled framework for resource allocation with explicit objective functions and constraints. In adversarial contexts, most works rely on stochastic, robust, or heuristic optimization models. For instance, Lin et al. \cite{lin2022heuristic} address edge computing resource allocation using a mixed-integer nonlinear program, offering a heuristic for energy-efficient offloading. Lian et al. \cite{lian2023distributed} propose distributed continuous-time algorithms for adversarially affected resource networks with projection-based convergence. In \cite{melouk2014stochastic}, discrete simulation optimization is used to address stochastic settings, supplemented by Tabu search and greedy heuristics. Li et al. \cite{li2021solving} model bi-objective stochastic problems but rely on known distributions and evolutionary solvers without convergence proofs. Plamondon et al. \cite{plamondon2006efficient} develop near-optimal real-time allocations using Markov Decision Processes with task dependencies. Recently, \cite{balseiro2023online} introduced online stochastic resource allocation under dynamic requests using dual mirror descent, offering regret bounds. Yaylali and Kalogerias \cite{yaylali2023robust} formulate a risk-aware allocation using Conditional Value-at-Risk to quantify robustness. Lotfi et al. \cite{lotfi2012centralized} transform stochastic centralized resource allocation into deterministic equivalents.

Despite the progress in this area, existing optimization-based approaches often suffer from one or more of the following limitations:
\begin{itemize}
    \item They assume known statistical distributions of attacks (e.g., known means or variances), which is rarely feasible in adversarial settings.
    \item They typically address either single-objective optimization (e.g., minimizing damage only) or use scalarization to combine objectives without addressing their inherent trade-offs.
    \item Many require integer programming or simulation-based heuristics that do not scale well and offer no theoretical guarantees.
\end{itemize}

In contrast, our work addresses a more general and practically relevant setting where the adversary’s attack probabilities are unknown and must be inferred implicitly through system feedback. We introduce a bi-objective optimization model that explicitly considers both cumulative damage and inter-node resource transfer costs. Our approach combines chance-constrained programming for uncertainty management with a network flow model for cost-efficient resource movement. Importantly, we provide a theoretical proof that our solution converges to the optimal solution corresponding to the true (unknown) attack means. This level of rigor and generality—handling unknown distributions, bi-objective structure, and scalable convergence—is not found in the existing literature.

\section{Problem Statement}

We consider a networked system of \(n\) computational nodes over a finite horizon of \(T\) discrete time slots, \(t = 1,2,\dots,T\). At each time \(t\), a total resource budget \(R\) must be allocated among the nodes and can be re–distributed via transfers. Meanwhile, an adversary may attack each node \(i\) with some unknown but stationary probability \(p_i = \Pr(\tilde y_i^t = 1)\). The goal of the problem of \emph{Resource Allocation in the Presence of Adversaries} (RAPA) is to design a dynamic allocation and transfer policy that first minimizes the cumulative damage expected caused by attacks and, secondly, minimizes the total cost of transferring resources between nodes.  

\subsection*{Notation}
\begin{table*}[!t]
\centering
\caption{Notation}
\label{tab:notation}
\footnotesize
\begin{tabular}{ll}
\toprule
\(\,n\) & Number of nodes in the network.\\
\(\,T\) & Number of time slots.\\
\(\,R\) & Total resource budget per time slot.\\
\(\,r_i^0\) & Initial resource at node \(i\).\\
\(\,r_i^t\) & Resource allocated to node \(i\) at time \(t\).\\
\(\,x_{ij}^t\) & Resource transferred from node \(i\) to node \(j\) at time \(t\).\\
\(\,w_i\) & Importance weight of node \(i\) (e.g., value or criticality).\\
\(\,c_{ij}\) & Cost per unit to transfer resource from \(i\) to \(j\).\\
\(\,\tilde y_i^t\in\{0,1\}\) & Random attack indicator on node \(i\) at time \(t\).\\
\(\,p_i\) & Unknown attack probability for node \(i\), \(p_i = \Pr(\tilde y_i^t = 1)\).\\
\(\,r_i^{\min},\,r_i^{\max}\) & Lower/upper thresholds for allocation to node \(i\).\\
\(\,D_i(r_i^t,\tilde y_i^t)\) & Damage function for node \(i\) at time \(t\).\\
\(\,Z_1,\,Z_2\) & Cumulative damage and total transfer cost objectives.\\
\bottomrule
\end{tabular}
\end{table*}

\subsection*{Decision Variables and Constraints}

At each time \(t\), nodes start with the allocation from the previous slot, re–distribute via \(x_{ij}^t\), and end with \(r_i^t\):
\[
r_i^t = r_i^{t-1} + \sum_{j=1}^n x_{ji}^t - \sum_{j=1}^n x_{ij}^t,
\quad\forall i,\,t\ge1.
\]
We enforce:
\begin{align}
&\sum_{i=1}^n r_i^t \le R, 
&& r_i^{\min} \le r_i^t \le r_i^{\max}, 
&& x_{ij}^t \ge 0,
\label{eq:budget-bounds}
\end{align}
for all \(i,j\in\{1,\dots,n\}\) and \(t=1,\dots,T\).  Here, \(r_i^{\min}>0\) ensures minimal operability, and \(r_i^{\max}\) is the point beyond which additional resource yields no further damage reduction.

\subsection*{Damage Function}
To capture the cost of under‐allocation in a way that is both intuitive and solver‐friendly, we adopt the following three‐segment piecewise‐linear damage function
\[
D_i:\mathbb{R}_{\ge0}\times\{0,1\}\;\longrightarrow\;[0,1]
\]
defined by
\begin{equation}
\label{eq:damage}
D_i(r_i^t,\tilde y_i^t)
=
\tilde y_i^t
\begin{cases}
1, 
& r_i^t \le r_i^{\min},\\[6pt]
\dfrac{r_i^{\max} - r_i^t}{r_i^{\max} - r_i^{\min}},
& r_i^{\min} < r_i^t < r_i^{\max},\\[8pt]
0, 
& r_i^t \ge r_i^{\max}.
\end{cases}
\end{equation}

\paragraph{Rationale and Precedent}  
This form:
\begin{itemize}
  \item Caps damage at 1 in the worst‐case shortfall (\(r_i^t\le r_i^{\min}\)),  
  \item Decays linearly as allocation rises toward the “safe” threshold \(r_i^{\max}\), and  
  \item Becomes zero once allocation suffices to cover the nominal demand.  
\end{itemize}
It appears repeatedly in facility location and network design under uncertainty.  For example, Snyder et al. reviewed comprehensively facility‐location under uncertainty and discussed this exact penalty for unmet demand \cite{snyder2006facility, snyder2005reliability}.
More generally, robust-optimization texts highlight that a three-segment piecewise-linear loss like balances modeling fidelity with tractability in MIP and QP solvers \cite{bertsimas2011theory,ShapiroDentchevaRuszczynski2009}.

\paragraph{Computational Tractability}  
Because there are only three linear pieces, modern solvers introduce just a few extra variables and constraints, yielding tight formulations and fast solve times.

\paragraph{Expectation under Uncertainty}  
If \(\tilde y_i^t\sim\text{Bernoulli}(p_i)\), then
\[
\mathbb{E}\bigl[D_i(r_i^t,\tilde y_i^t)\bigr]
= 
\begin{cases}
p_i , & r_i^t \le r_i^{\min},\\
p_i \times \dfrac{r_i^{\max} - r_i^t}{r_i^{\max} - r_i^{\min}}, & r_i^{\min} < r_i^t < r_i^{\max},\\
0,& r_i^t \ge r_i^{\max}.
\end{cases}
\]
This closed‐form makes integration into risk‐averse objectives straightforward.

\subsection*{Bi-Objective Formulation}

We seek to minimize the following.
\begin{align}
Z_1 &= \sum_{t=1}^T \sum_{i=1}^n w_i \,\mathbb{E}\bigl[D_i(r_i^t,\tilde y_i^t)\bigr],
\label{eq:obj-damage}\\
Z_2 &= \sum_{t=1}^T \sum_{i=1}^n \sum_{j=1}^n c_{ij}\,x_{ij}^t.
\label{eq:obj-cost}
\end{align}

Subject to the feasibility constraints in \eqref{eq:budget-bounds} and the flow–balance constraint above.

\subsection*{Complexity and Solution Approach}

RAPA has \(O(n^2T)\) decision variables and \(O(nT)\) constraints, with unknown probabilities \(p_i\).  Enumerating the Pareto front is intractable—even under full information—due to potentially infinite or exponentially many trade-offs.  Instead, we:

\begin{enumerate}
  \item Reformulate the damage objective via a chance–constraint surrogate (see the next Section).
  \item Decompose into two interlinked single–objective subproblems.
  \item Introduce a tunable risk parameter \(\alpha\) to trace efficient solutions without full Pareto enumeration.
\end{enumerate}


\section{Decomposition Approach}

In RAPA, the damage‐minimization objective \(Z_{1}\) depends on the random attack indicators \(\tilde y_i^t\), whereas the transfer‐cost objective \(Z_{2}\) depends only on the transfer variables \(x_{ij}^t\).  We exploit this structure by decomposing the bi‐objective problem into two coupled subproblems solved sequentially at each time \(t\).  This decomposition both simplifies computation and preserves the interaction between allocation and transfer decisions.

\subsection*{Overview of the Two‐Phase Method}

At each time slot \(t\), given estimates of the attack distribution parameters \(\mathbb{E}[\tilde y_i^t]\) and \(\mathrm{Var}[\tilde y_i^t]\):
\begin{enumerate}
  \item[(I)]  \textbf{Allocation Phase:} solve a chance‐constrained program to choose \(r_i^t\) so that with confidence \(1-\alpha\) the expected damage does not exceed a threshold \(\epsilon^t\).  This yields \(\{r_i^t\}^*\).
  \item[(II)] \textbf{Transfer Phase:} solve a minimum‐cost network‐flow problem to find \(x_{ij}^t\) that implements the allocation \(\{r_i^t\}^*\) from the previous slot’s ending inventories.
\end{enumerate}

The two phases are linked because the transfer phase must deliver the allocations computed in the allocation phase, and the observed attacks after transfer feed back to update the attack‐distribution estimates for the next slot.

\subsection*{Phase I: Chance‐Constrained Allocation}

We introduce an auxiliary variable \(\epsilon^t\) that upper‐bounds the total weighted damage in slot \(t\).  The chance‐constraint

\[
\Pr\Bigl(\sum_{i=1}^n w_i\,D_i(r_i^t,\tilde y_i^t)\;\le\;\epsilon^t\Bigr)
\;\ge\;1-\alpha ,
\]

where $\Pr(.)$ is the probability with respect to the random attack indicators \(\tilde y_i^t\). This constraint ensures that, with probability at least \(1-\alpha\), damage does not exceed \(\epsilon^t\).  Recalling
\[
D_i(r_i^t,\tilde y_i^t)
=\tilde y_i^t\,
\frac{r_i^{\max}-r_i^t}{r_i^{\max}-r_i^{\min}},
\]
and defining 
\[
v_i^t \;=\; w_i\,\frac{r_i^{\max}-r_i^t}{r_i^{\max}-r_i^{\min}},
\]
a standard Cantelli‐type bound \cite{nikolova2010approximation,yang2017algorithm} yields the deterministic surrogate

\begin{equation}
\sum_{i=1}^n \mathbb{E}[\tilde y_i^t]\,v_i^t
\;+\;\sqrt{\tfrac{1-\alpha}{\alpha}}\,
\sqrt{\sum_{i=1}^n \mathrm{Var}[\tilde y_i^t]\,v_i^t}
\;\le\;\epsilon^t.
\label{eq:deterministic_cc}
\end{equation}

Hence, \textbf{Subproblem I} at time \(t\) is
\begin{align}
\min_{r_i^t,\;\epsilon^t}\quad & \epsilon^t \\
\text{s.t.}\quad
& \eqref{eq:deterministic_cc}, \nonumber\\
& \sum_{i=1}^n r_i^t \;\le\; R, \quad
r_i^{\min}\;\le\;r_i^t\;\le\;r_i^{\max},\;\forall i.
\nonumber
\end{align}

\textbf{Remark 1.} The risk parameter \(\alpha \in [0,1]\) in the chance‐constraint (Phase I) acts as an explicit trade‐off knob between damage minimization (\(Z_1\)) and transfer‐cost minimization (\(Z_2\)).  
\begin{itemize}
  \item When \(\alpha\to 0\), the constraint in \eqref{eq:deterministic_cc} becomes more conservative, forcing the allocation to guard against high‐damage outcomes almost surely; this yields a risk‐averse solution that may incur higher transfer costs.  
  \item When \(\alpha\to 1\), the chance‐constraint is relaxed, allowing higher damage risk in exchange for lower expected transfer cost—i.e., a more risk‐seeking strategy.  
  \item By continuously varying \(\alpha\) between 0 and 1, one traces out a continuum of “practical” Pareto‐optimal solutions reflecting different risk–cost trade‐offs.  
\end{itemize}
Equivalently, one can view \(\alpha\) as defining the weight in a scalarized objective
\[
\min\;\lambda(\alpha)\,Z_1 \;+\;\bigl(1-\lambda(\alpha)\bigr)\,Z_2,
\]
where \(\lambda(\alpha)\) is a monotonically decreasing function of \(\alpha\).  This interpretation clarifies that tuning \(\alpha\) is akin to selecting a point on the Pareto front under uncertainty, balancing robustness (low damage risk) against efficiency (low transfer cost).

\subsection*{Phase II: Minimum‐Cost Network Flow}

Let \(\hat r_i^{t-1}\) be the post‐transfer inventory of node \(i\) at the end of slot \(t-1\).  Having computed \(r_i^t{}^*\) in Phase~I, we must transfer resource so that each node achieves \(r_i^t{}^*\).  Define supply/demand at node \(i\):
\[
s_i \;=\;\hat r_i^{t-1} - r_i^t{}^*\quad
\begin{cases}
>0 &: \text{surplus at }i,\\
<0 &: \text{deficit at }i.
\end{cases}
\]
Introduce a directed graph with nodes \(1,\dots,n\) plus super‐source and super‐sink.  Connect source to each surplus node \(i\) with capacity \(s_i\), cost \(0\); connect each deficit node \(j\) to sink with capacity \(-s_j\), cost \(0\); and for each \(i,j\) include edge \((i,j)\) with capacity \(R\) and cost \(c_{ij}\).  

\textbf{Subproblem II} is the standard minimum‐cost flow:

\begin{align}
\min_{x_{ij}^t}\quad & \sum_{i,j} c_{ij}\,x_{ij}^t \\
\text{s.t.}\quad
& \sum_j x_{ij}^t - \sum_j x_{ji}^t = s_i,\;\forall i, \\
& 0\;\le\;x_{ij}^t\;\le\;R,\;\forall i,j.
\nonumber
\end{align}

This LP runs in polynomial time (e.g.\ via successive shortest‐path or network simplex).

\begin{algorithm}[t]
\caption{Two‐Phase RAPA Solver}
\label{alg:arap}
\begin{algorithmic}[1]
\STATE Initialize \(\mathbb{E}[\tilde y_i^1]\), \(\mathrm{Var}[\tilde y_i^1]\), and \(r_i^0\).
\FOR{\(t=1\) to \(T\)}
  \STATE \emph{(Phase I)} Solve Subproblem I \(\to\{r_i^t\}^*, \epsilon^t\).
  \STATE \emph{(Phase II)} Solve Subproblem II \(\to\{x_{ij}^t\}^*\).
  \STATE Apply transfers, observe attacks \(\tilde y_i^t\).
  \STATE Update \(\mathbb{E}[\tilde y_i^{t+1}], \mathrm{Var}[\tilde y_i^{t+1}]\) (e.g.\ Bayesian or sample averages).
\ENDFOR
\STATE \textbf{return} \(\{r_i^t\}^*,\{x_{ij}^t\}^*\).
\end{algorithmic}
\end{algorithm}

The details of the proposed approach to solve RAPA is presented in the Appendix.

 \subsection{Convergence, Complexity and Optimality Analysis}
\begin{theorem}
    For a specific risk parameter $\alpha$, as the number of iterations $t$ increases, the solution obtained by the RAPA algorithm tends to the true optimal solution of the resource allocation problem in time slot $t$.
\end{theorem}
\begin{proof}
The attack distribution is unknown but remains unchanged over time. Suppose $\mu$ and $\sigma^2$ are the true mean and true variance of the attack distribution. We can write
\begin{align*}
    & \quad \mathbb{E}[\tilde{y}_i^{t}] = \mu+ \Delta_{\text{mean}}^t,\\
    & \quad \mathbb{V}[\tilde{y}_i^{t}] = \sigma^2 + \Delta_{\text{Var}}^t,
\end{align*}
where $\Delta_{\text{mean}}^t$ and $\Delta_{\text{Var}}^t$ are the estimation mean error and variance error in time slot $t$. All the constraints and the objective in the model are deterministic except for the probability of attacks. Substituting the estimated error in this constraint yields
\begin{align*}
\sum_{i=1}^{n} (\mu + \Delta_{\text{mean}}^t ) v_i^t + \sqrt{\frac{1-\alpha}{\alpha}} \sqrt{\sum_{i=1}^{n} (\sigma^2 + \Delta_{\text{Var}}^t) v_i^t} \leq \epsilon^t.
\end{align*}
By applying the strong law of large numbers (Kolmogorov's law) \cite{chung2008strong}, we know that $\Delta_{\text{mean}}^t \to 0$ and $\Delta_{\text{Var}}^t \to 0$ as $T \to \infty$. That is, the estimation error in both the mean and variance decreases over iterations. Formally,
\begin{align*}
\lim_{t \to \infty} \sum_{i=1}^{n} \mathbb{E}[\tilde{y}_i^{t}] v_i^t + \sqrt{\frac{1-\alpha}{\alpha}} \sqrt{\sum_{i=1}^{n} \mathbb{V}[\tilde{y}_i^{t}] v_i^t} = \\ \sum_{i=1}^{n} \mu v_i^t + \sqrt{\frac{1-\alpha}{\alpha}} \sqrt{\sum_{i=1}^{n} \sigma^2 v_i^t}.
\end{align*}
Thus, the model and the constraints approach a deterministic form devoid of error or uncertainty. Consequently, the solution derived from the RAPA algorithm converges towards the true optimal solution.
\end{proof}

Since by increasing $t$ the obtained solution by RAPA approaches the optimal solution in time slot $t$, the sum of objective values, $\Sigma_{t=1}^T \epsilon^t$, approaches the optimal $Z_1$ value.

The proposed decomposition approach is an iterative two-step algorithm. In the first step, given \(\alpha\), the algorithm minimizes the damage function. In the second step, it computes the optimal resource transfer. Since there is no dependency or limit between allocated resources over time and nodes, i.e., the variables \(r_i^t\) and \(r_j^{t+1}\) for any \(i, j\), and \(t\), minimizing the damage independently in each time slot minimizes the overall damage. As a result, the obtained solution is non-dominated, meaning no other solution achieves a lower level of damage. Therefore, it is the extreme Pareto-optimal solution with the minimum damage. 

The algorithm solves a \textit{quadratic model} to find the optimal resource allocation in the first step and a \textit{linear model} to find the optimal resource transfer in the second step. We applied Gurobi version 11.0.1 \cite{gurobi2024} and SciPy version 1.13.0 \cite{SciPy-NMeth2020} in the first and the second steps, respectively. Gurobi employs advanced algorithms such as the primal-dual interior-point method \cite{nesterov1994interior} to solve quadratic programming. The theoretical complexity for the interior-point method is polynomial, specifically \(O(n^{3.5}L)\), where \(n\) is the number of variables and \(L\) is the input size in bits. Empirically, Gurobi demonstrates robust performance, efficiently solving large instances, indicating effective optimization beyond theoretical worst-case bounds. The network flow problem is solved using the SciPy library, which employs various efficient algorithms such as Ford-Fulkerson (Edmonds-Karp) \cite{edmonds1972theoretical}, Dinic's \cite{dinic1970algorithm}, and Push-Relabel \cite{goldberg1988new}. Theoretical complexities for these algorithms are \(O(n^5)\), \(O(n^4)\), and \(O(n^3)\) respectively. However, SciPy indicates its implementation is highly optimized and performs well in practice. This suggests that while theoretical worst-case complexities provide an upper bound, the actual performance of SciPy is efficient and suitable for RAPA.

\section{Simulation Results}
We conduct comprehensive simulations to evaluate our proposed approach, called \textbf{Un-mean}, to solve RAPA across various scenarios. We also compare it with the following optimization approaches:
\begin{itemize}
  \item \textbf{Kn-mean}: A variant of our method that assumes \emph{a priori} knowledge of the true mean and variance of node–attack probabilities.  By removing the online update step, it solves the chance‐constrained allocation with exact statistics.  Comparing against Kn-mean quantifies how quickly Un-mean learns and approaches the performance of a fully informed method.
  \item \textbf{Greedy}: A lightweight heuristic that first guarantees each node’s minimum resource \(r_i^{\min}\), then distributes the remaining budget proportionally according to a combined score of node importance and empirical attack frequency.  This baseline evaluates the benefit of our principled chance‐constraint approach versus a simple, data‐driven rule.
  \item \textbf{Oracle}: An “omniscient” benchmark that knows each node’s actual attack outcome in the upcoming time slot.  It solves a deterministic LP each round to minimize damage, without regard for statistical uncertainty.  Oracle defines the lower bound on damage and illustrates the cost of uncertainty in transfer-phase decisions.
\end{itemize}

All three baselines go beyond a naive greedy rule—Kn-mean uses full statistical knowledge in the same chance‐constrained model, and Oracle solves a dynamic LP with perfect foresight—providing a spectrum from heuristic to fully informed strategies.

We implemented these approaches in Python, using Gurobi \cite{gurobi2024} for the chance-constrained and LP allocations, and SciPy \cite{SciPy-NMeth2020} for the network flow. Experiments ran on an Intel i7-11800H CPU with 16 GB RAM.

We generate each RAPA instance by creating a network of \(n\) nodes with randomly assigned importance weights (\(w_i\)) and transfer costs (\(c_{ij}\)). Additionally, we randomly generate damage functions \(D\) with parameters \(r_i^{\text{min}}\) and \(r_i^{\text{max}}\) for \(i = 1,2,\ldots,n\). We also determine \(T\) and a random total available resource $R$ between $\sum_{i} r_i^{\text{min}}$ and $\sum_{i} r_i^{\text{max}}$. Subsequently, we apply a Bernoulli distribution with a random attack probability \(p\) to produce \(T\) attack possibilities for each node. 
\subsection{Method Comparison}
For comparison, we selected four pairs of \((n,T)\): \((30,20)\), \((30,80)\), \((100,20)\), and \((100,80)\) to plot the result. We executed each algorithm on 50 different random scenarios, varying the importance weights, transfer costs, resource availability, damage functions \(D\), and attack probabilities. Figure \ref{fig_comparison} shows the results. The left-hand panels display the obtained average damage for each approach, while the right-hand panels show the resource transfer cost. \\
As expected, the Oracle consistently achieves the minimum damage values because it has prior knowledge of the attacks. Consequently, Oracle optimally allocates resources \(R\) to only the nodes that will be attacked in the next time slot. However, this strategy necessitates transferring substantial resources in the network. As a result, Oracle has the highest transfer cost compared to the other three approaches. In contrast, the Kn-mean approach distributes resources based on known mean and variance. This results in the minimum transfer cost compared to the other approaches. The Greedy approach performs the worst in terms of damage but achieves a better transfer cost than Oracle. The Un-mean approach outperforms Greedy in terms of damage but falls short of Kn-mean and Oracle due to the information shortage, while achieving better transfer costs than Oracle and Greedy. As the number of time slots increases from 20 to 80, Un-mean converges towards Kn-mean as a result of learning, so that the adverse effect of the information absence vanishes. 

Concerning the running time, the Greedy approach is the fastest, followed by Oracle, Kn-mean, and Un-mean. For instance, in the \((100,80)\) scenario, Greedy takes $2.3$ seconds, Oracle $4.4$, Kn-mean $5.6$, and Un-mean $5.8$. This ranking is due to the computational complexity of each method: All methods execute two phases in each round: Finding a resource allocation and computing the optimal resource transfer. The second phase is identical across all approaches, involving the execution of the described network flow method, which is a linear program with $(O(n^2))$ variables to compute $(x_{ij}^t)^*$. Thus, the main distinction lies in the first phase:  In each time slot $t$, both the Kn-mean and Un-mean solve a quadratic program with $O(n)$ variables to determine $(r_i^t)^*$ within the chance-constrained model. Besides, Un-mean requires updating the mean and variance of attacks. Since the Oracle knows the attack probabilities, it solves a deterministic linear program with $O(n)$ variables to find the optimal resource allocations $(r_i^t)^*$ in the first phase. Finally, the Greedy approach allocates resources using a simple proportional resource share based on the weights of importance and frequency of attacks. So, it does not need to solve a linear program.
\begin{figure*}[!t]
    \centering
    \begin{subfigure}{0.24\textwidth}
        \centering
        \includegraphics[width=\linewidth]{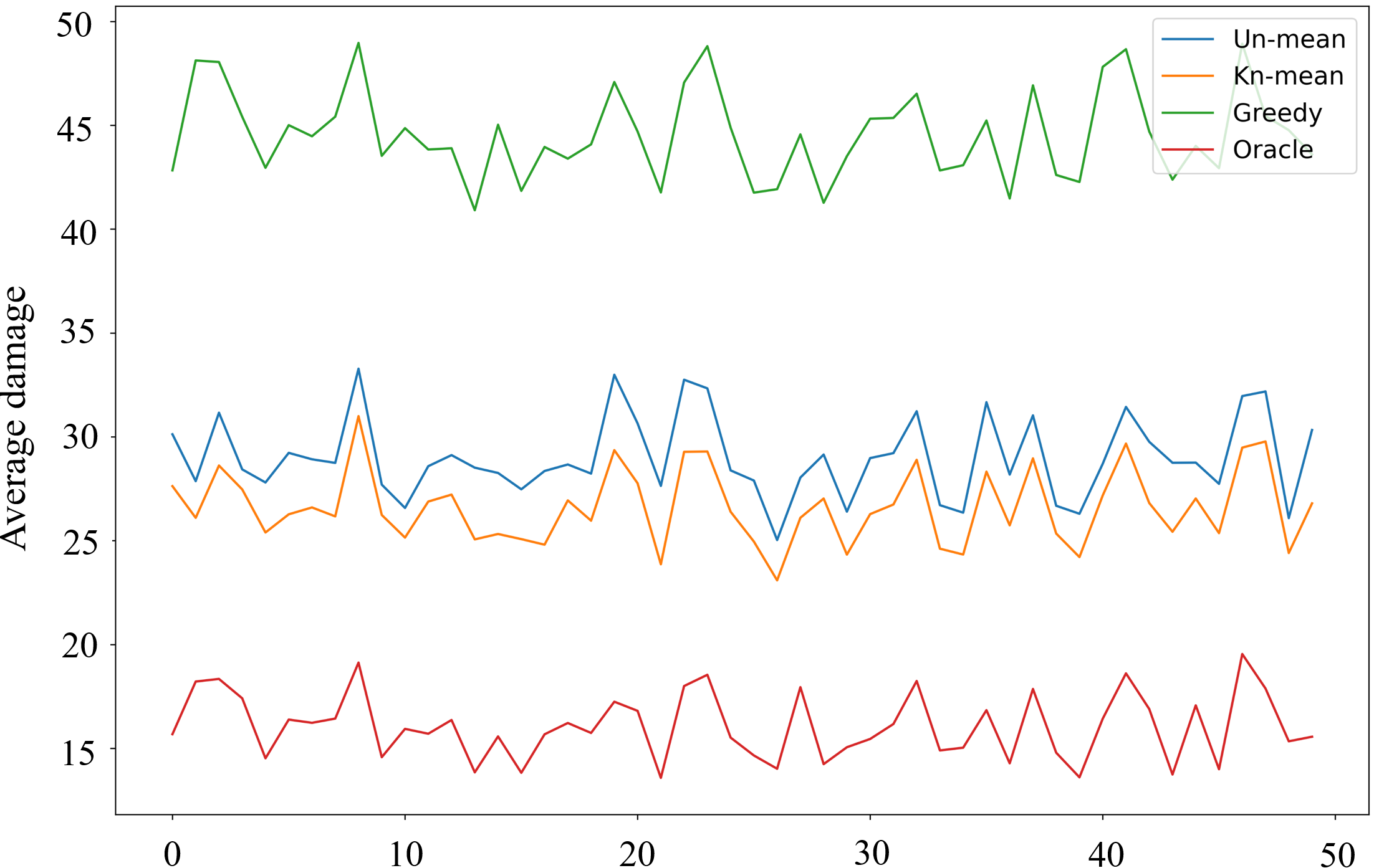}
        \caption*{Damage values for ($30, 20$)}
    \end{subfigure}
    \hfill
    \begin{subfigure}{0.24\textwidth}
        \centering
        \includegraphics[width=\linewidth]{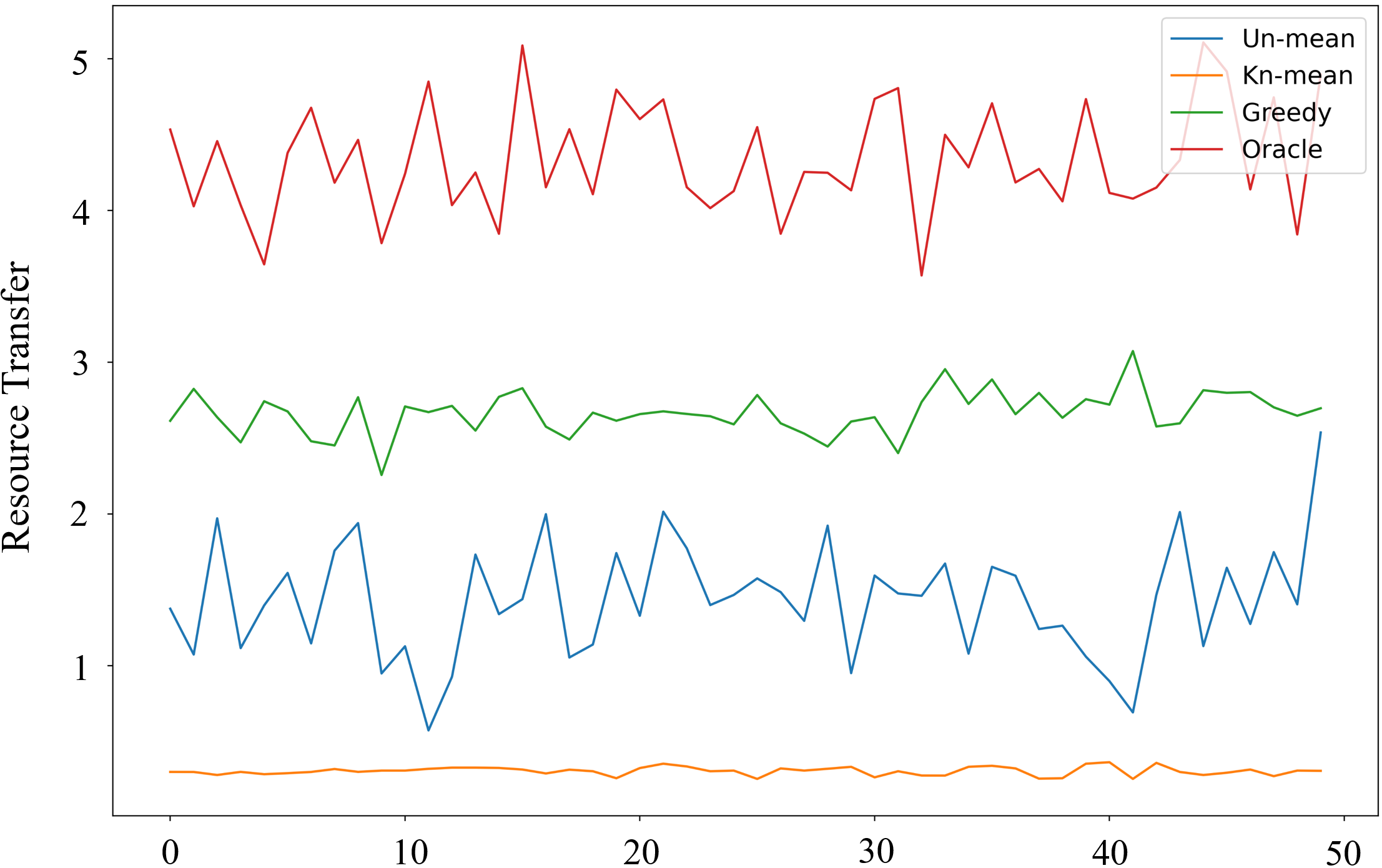}
        \caption*{Transfer cost for ($30, 20$)}
    \end{subfigure}
    \hfill
    \begin{subfigure}{0.24\textwidth}
        \centering
        \includegraphics[width=\linewidth]{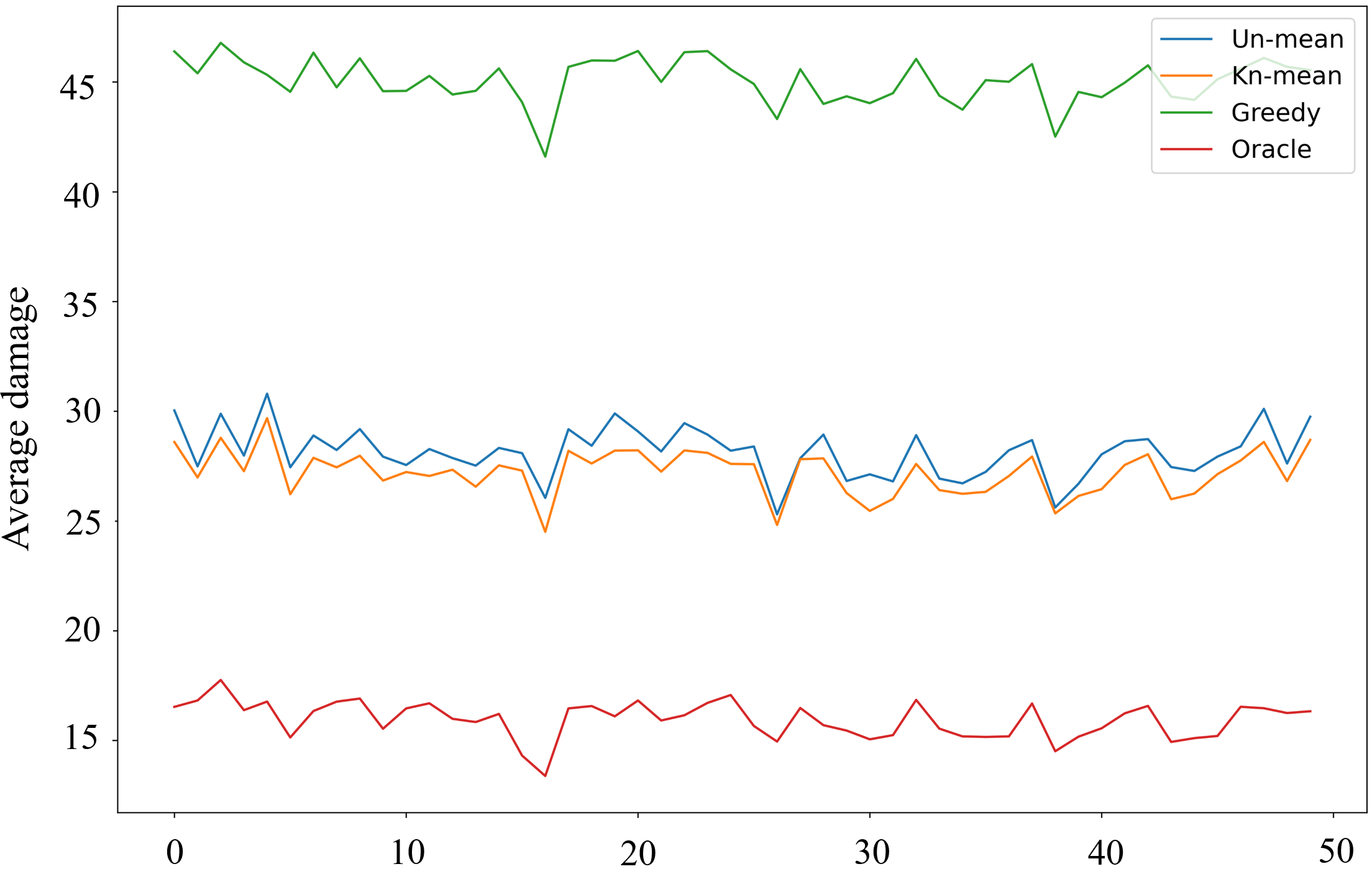}
        \caption*{Damage values for ($30, 80$)}
    \end{subfigure}
    \hfill
    \begin{subfigure}{0.24\textwidth}
        \centering
        \includegraphics[width=\linewidth]{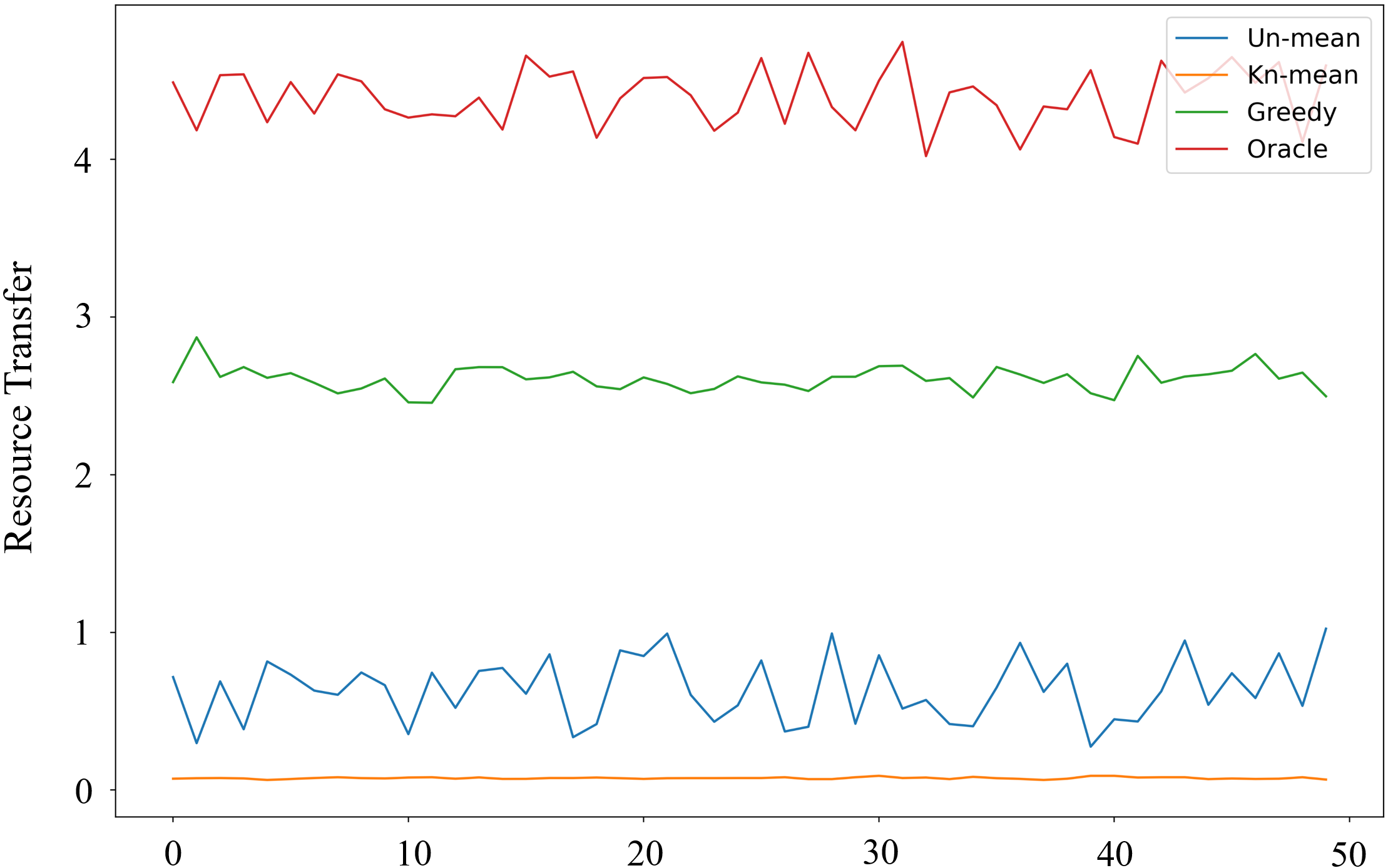}
        \caption*{Transfer cost for ($30, 80$)}
    \end{subfigure}

    \begin{subfigure}{0.24\textwidth}
        \centering
        \includegraphics[width=\linewidth]{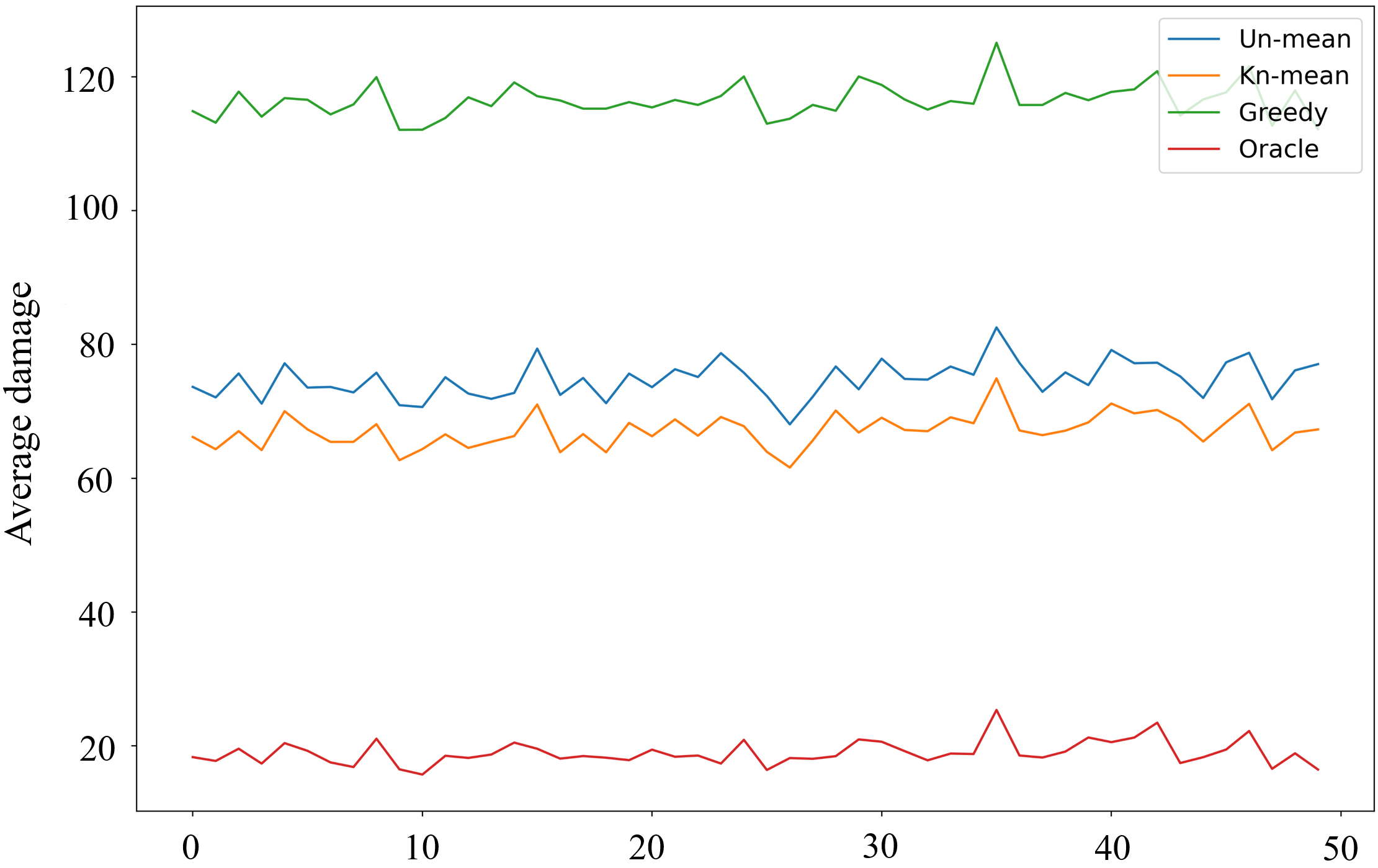}
        \caption*{Damage values for ($100, 20$)}
    \end{subfigure}
    \hfill
    \begin{subfigure}{0.24\textwidth}
        \centering
        \includegraphics[width=\linewidth]{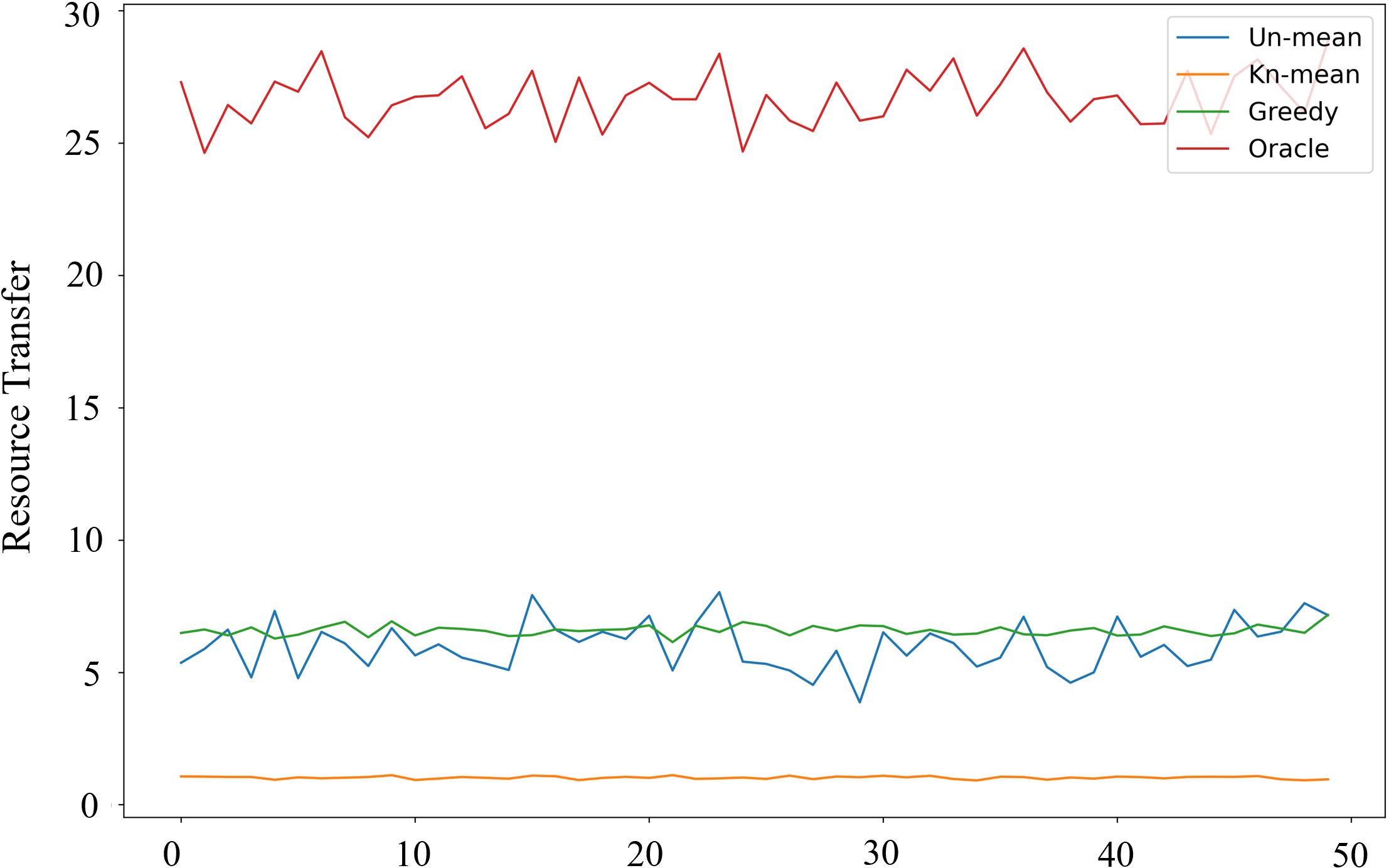}
        \caption*{Transfer cost for ($100, 20$)}
    \end{subfigure}
    \hfill
    \begin{subfigure}{0.24\textwidth}
        \centering
        \includegraphics[width=\linewidth]{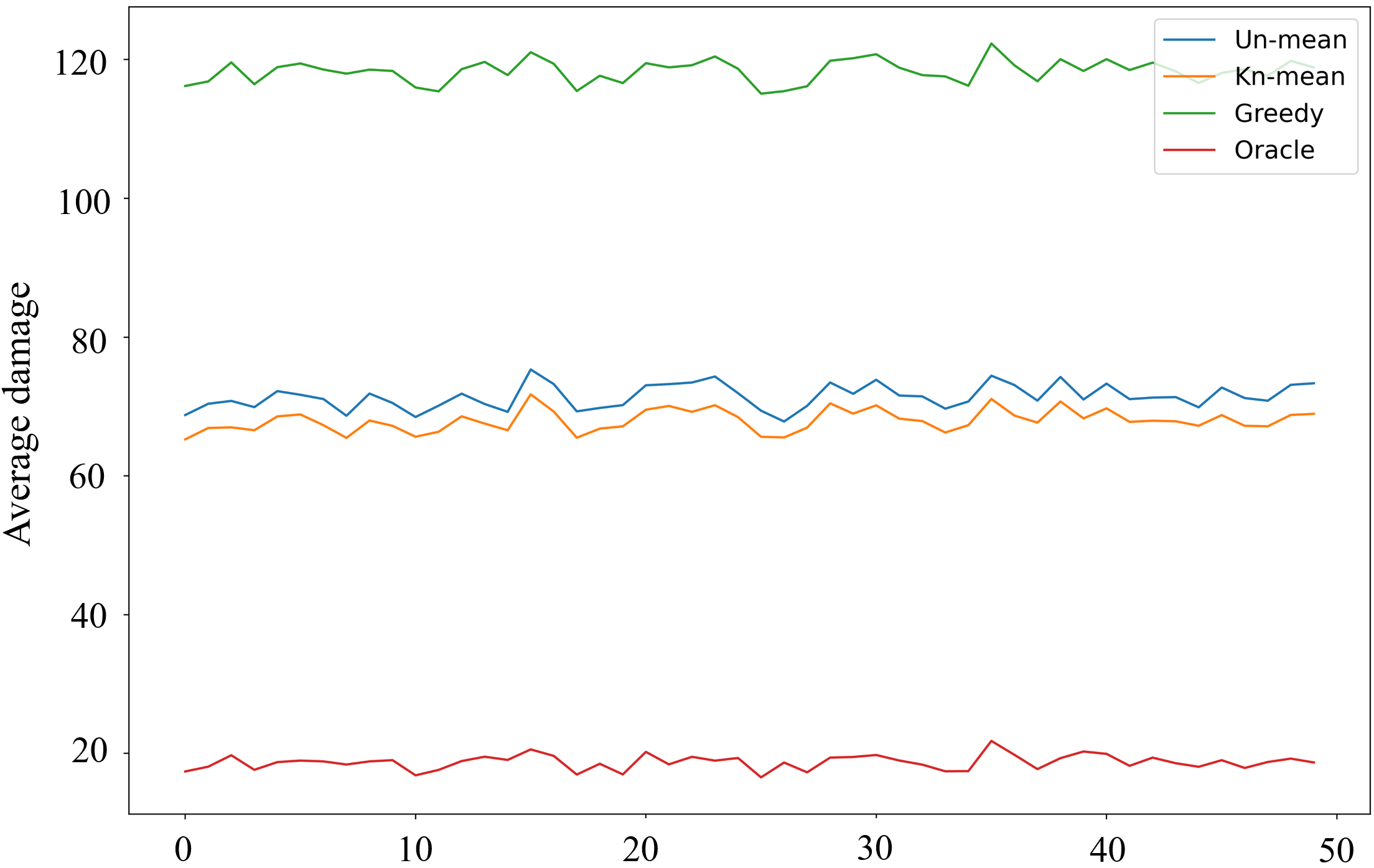}
        \caption*{Damage values for ($100, 80$)}
    \end{subfigure}
    \hfill
    \begin{subfigure}{0.24\textwidth}
        \centering
        \includegraphics[width=\linewidth]{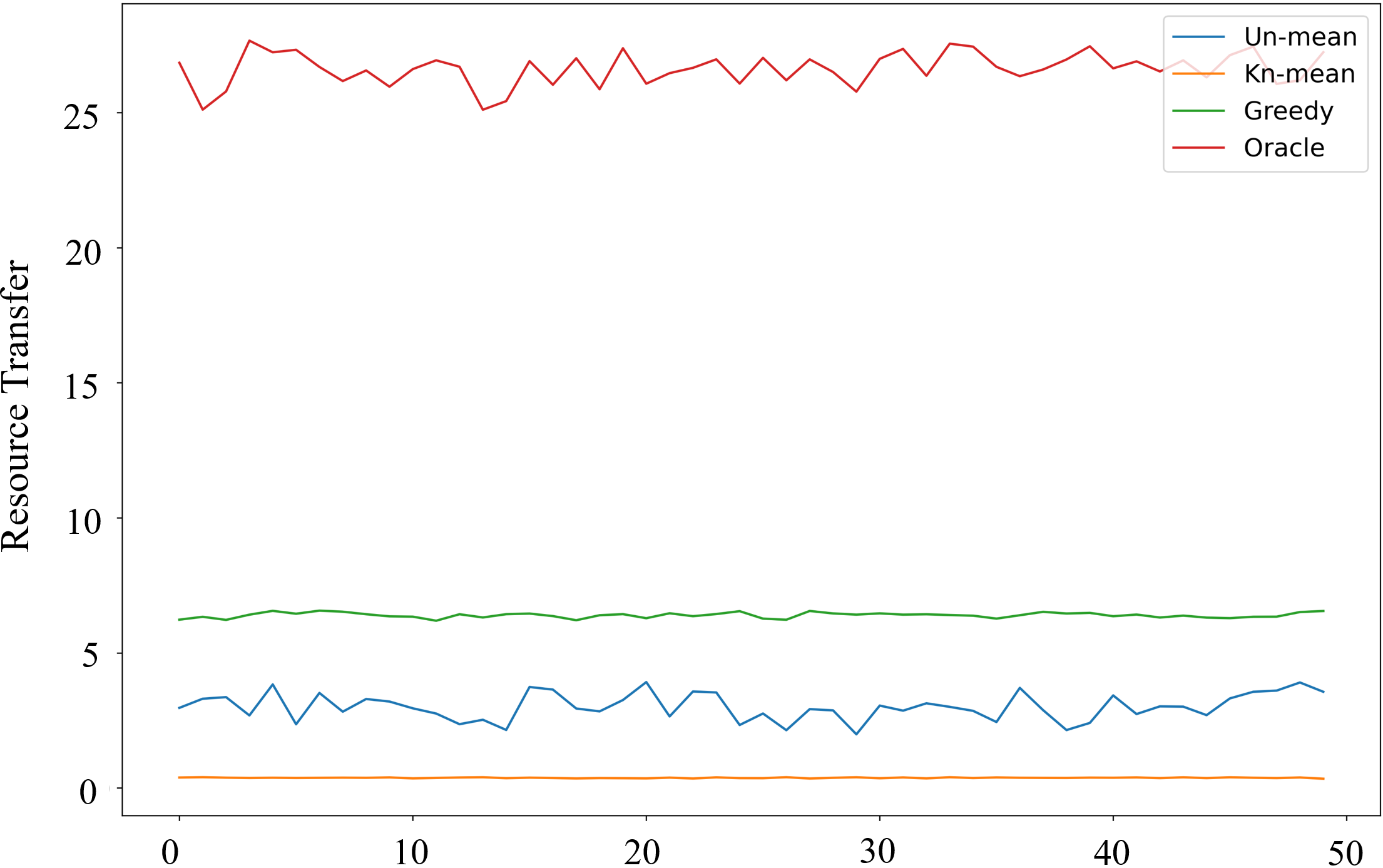}
        \caption*{Transfer cost for ($100, 80$)}
    \end{subfigure}

    \caption{Average damage and transfer cost obtained by four approaches Un-mean, Kn-mean, Greedy, and Oracle for 50 instances (50 different scenarios on horizontal axis) of RAPA, denoted by $(n,T)$.}
    \label{fig_comparison}
\end{figure*}

In addition to the visual comparison results, Table \ref{comparison_table} presents the damage value, transfer cost, and running time for each approach. Each result is based on the average of 30 random instances of RAPA, conducted under varying numbers of nodes, $n=25,50,100,200$, and time slots $T=20,40,80$. The Greedy approach is the fastest, followed by Oracle, Kn-mean, and Un-mean, in that order. Regarding the damage function, the Greedy approach demonstrates the worst performance, while the Oracle has the best performance, as expected. Un-mean and Kn-mean fall between these two extremes. Additionally, the damage value obtained by Un-mean approaches that of Kn-mean as the number of time slots, $T$, increases. Finally, in terms of transfer cost, Oracle exhibits the worst performance, whereas Kn-mean demonstrates the best performance, as expected.
\begin{table*}[!t]
\caption{The damage and transfer cost (\textit{tran} column), and running time (\textit{time} column) of Un-mean, Kn-mean, Greedy and Oracle approaches over 12 different settings of RAPA. The results are the average of 30 random scenarios.}
\label{comparison_table}
\centering
\begin{tabular}{clccccccccccccc}
\hline
                     &         &  & \multicolumn{3}{c}{\textit{T}=20} &  & \multicolumn{3}{c}{\textit{T}=40} &  & \multicolumn{3}{c}{\textit{T}=80} \\ \cline{4-6} \cline{8-10} \cline{12-14}
\textit{n}           &         &  & damage  & tran & time &  & damage  & tran & time &  & damage  & tran & time \\ \hline
\multirow{4}{*}{25}  & Un-mean &  & 17.78 & 1.84 & 0.40 &  & 19.69 & 0.83 & 0.79 &  & 13.35 & 0.49 & 1.96 \\
                     & Kn-mean &  & 14.64 & 0.37 & 0.42 &  & 18.24 & 0.13 & 0.59 &  & 12.68 & 0.07 & 1.90 \\
                     & Greedy  &  & 25.67 & 2.94 & 0.06 &  & 28.60 & 2.56 & 0.12 &  & 24.07 & 2.52 & 0.32 \\
                     & Oracle  &  & 3.73  & 5.70 & 0.24 &  & 7.20  & 5.53 & 0.45 &  & 3.83  & 5.16 & 1.14 \\ \hline
\multirow{4}{*}{50}  & Un-mean &  & 38.08 & 3.61 & 1.09 &  & 33.78 & 1.63 & 2.01 &  & 27.86 & 1.22 & 3.60 \\
                     & Kn-mean &  & 32.48 & 0.66 & 1.07 &  & 31.20 & 0.34 & 1.95 &  & 26.54 & 0.16 & 3.20 \\
                     & Greedy  &  & 55.98 & 4.15 & 0.34 &  & 56.21 & 4.22 & 0.70 &  & 47.74 & 3.82 & 1.32 \\
                     & Oracle  &  & 11.60 & 11.82 & 0.65 &  & 11.71 & 10.89 & 1.42 &  & 8.14  & 10.84 & 2.90 \\ \hline
\multirow{4}{*}{100} & Un-mean &  & 72.04 & 5.65 & 1.58 &  & 66.18 & 3.97 & 3.65 &  & 64.05 & 2.38 & 5.82 \\
                     & Kn-mean &  & 62.60 & 1.37 & 1.50 &  & 60.59 & 0.72 & 3.49 &  & 60.85 & 0.30 & 5.63 \\
                     & Greedy  &  & 106.42 & 6.00 & 0.86 &  & 106.51 & 6.13 & 1.48 &  & 106.08 & 5.20 & 2.30 \\
                     & Oracle  &  & 21.09  & 22.42 & 1.34 &  & 17.84  & 23.10 & 2.23 &  & 18.52  & 20.88 & 4.38 \\ \hline
\multirow{4}{*}{200} & Un-mean &  & 149.64 & 11.54 & 9.21 &  & 125.11 & 7.83 & 16.84 &  & 128.51 & 3.97 & 29.86 \\
                     & Kn-mean &  & 130.55 & 2.57  & 8.91 &  & 112.39 & 1.29 & 15.27 &  & 123.00 & 0.64 & 28.61 \\
                     & Greedy  &  & 220.05 & 8.55  & 5.06 &  & 199.54 & 8.41 & 9.14  &  & 218.48 & 7.92 & 11.27 \\
                     & Oracle  &  & 53.48  & 39.32 & 7.38 &  & 29.16  & 42.30 & 14.20 &  & 40.05  & 42.35 & 25.70 \\ \hline
\end{tabular}
\end{table*}
\subsection{Learning Ability of the Algorithm}
An essential requirement of our proposed solution is to learn the attack scenarios over time. To evaluate this feature, we compare it to the Kn-mean algorithm that knows the mean and variance in advance. Figure \ref{fig_learn} depicts the results for two different instances of RAPA, ($n=50, T=100$) and ($n=100, T=100$). It shows that the difference between the damage values from the two algorithms decreases over time and gradually converges to zero. Thus, the algorithm effectively approximates the true mean and variance of the attack scenarios with time. The convergence towards zero demonstrates the learning robustness.
\begin{figure}[!t]
    \centering
    \begin{subfigure}{0.48\linewidth}
        \centering
        \includegraphics[width=\linewidth]{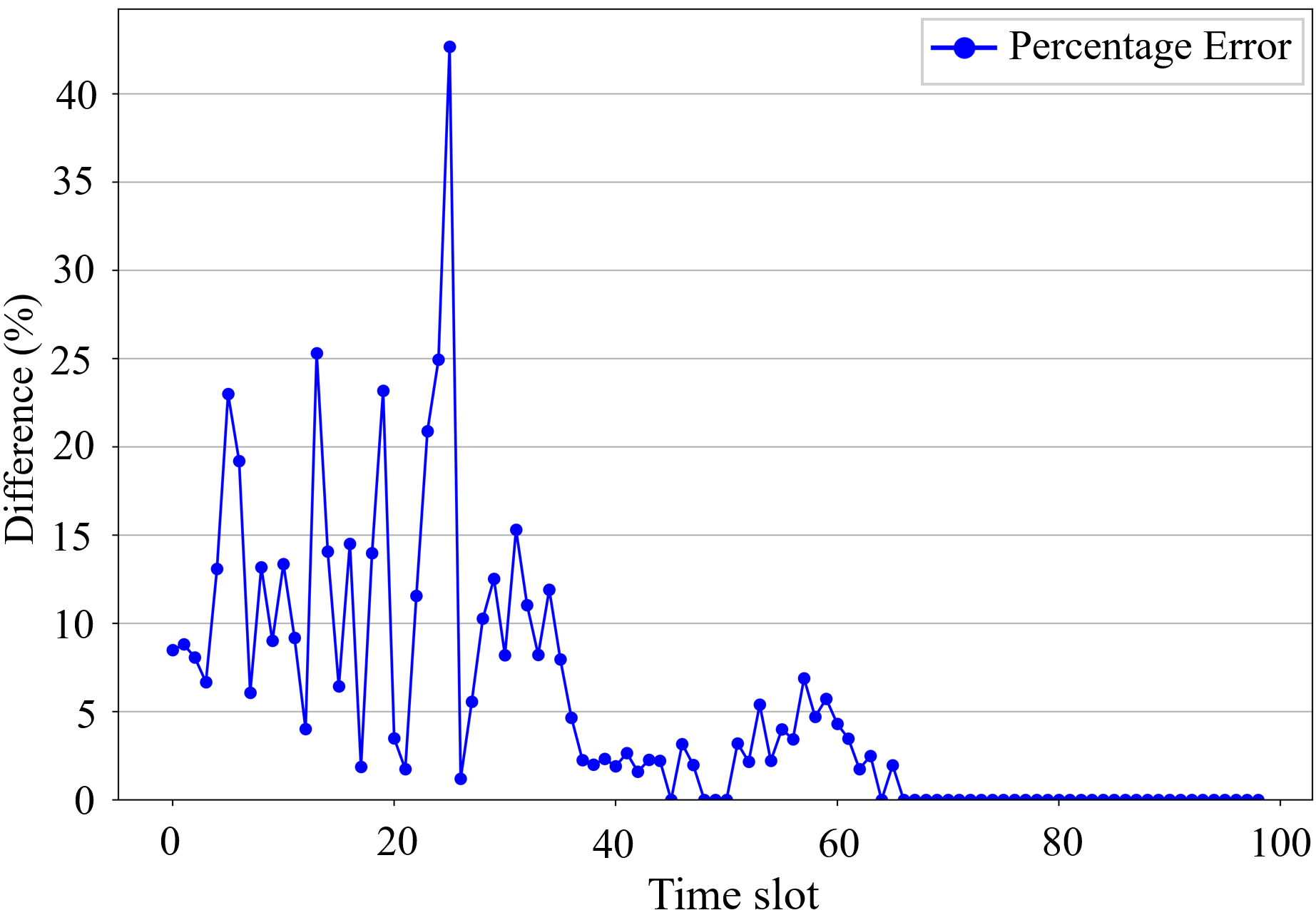}
    \end{subfigure}
    \hfill
    \begin{subfigure}{0.48\linewidth}
        \centering
        \includegraphics[width=\linewidth]{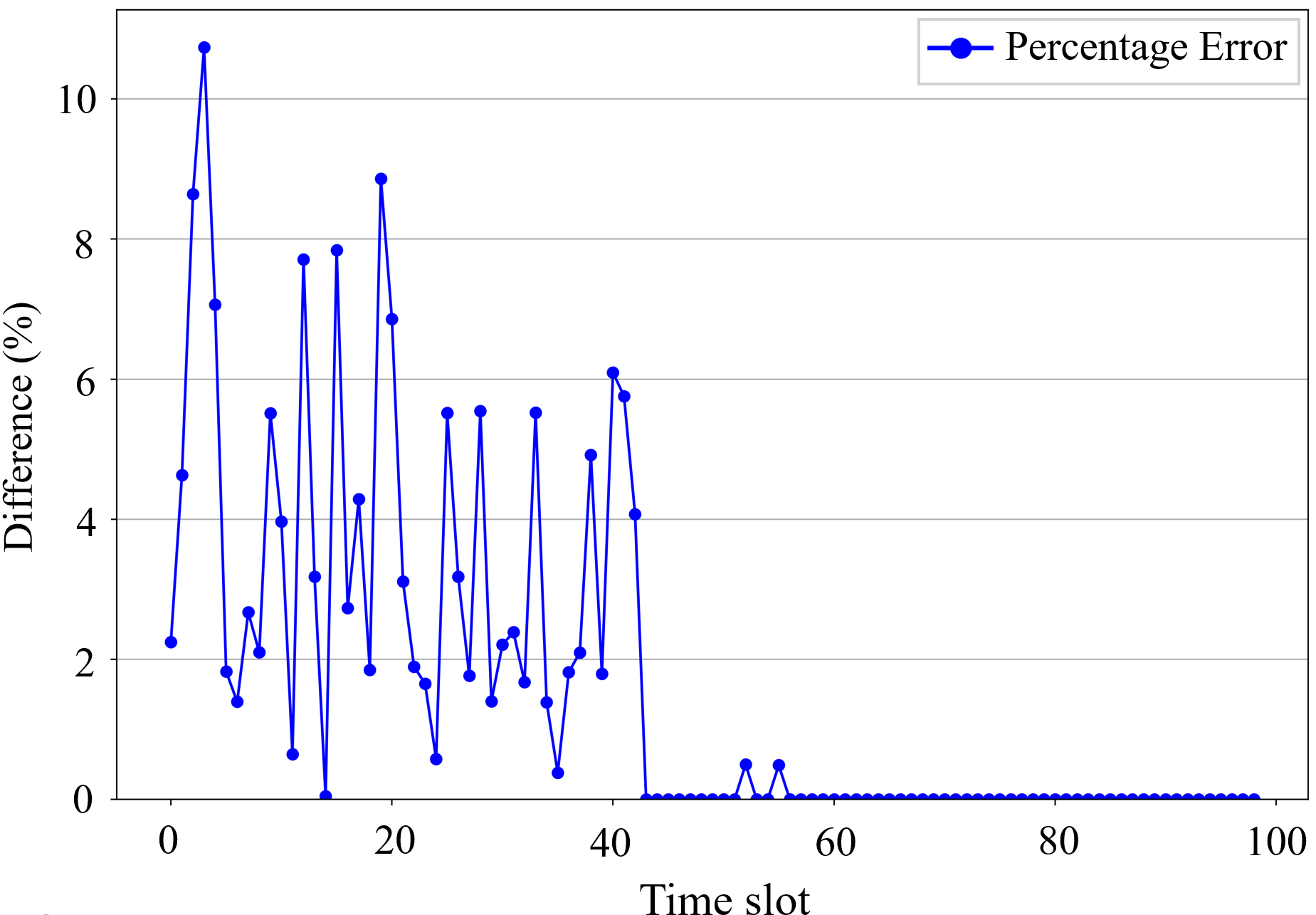}
    \end{subfigure}
    \caption{Learning ability of Un-mean algorithm. The diagram shows the difference between the obtained damage by Un-mean and Kn-mean algorithms over time for the instances ($50, 100$) (left panel), and ($100, 100$) (right panel).}
    \label{fig_learn}
\end{figure}

\subsection{Effect of Risk Parameter \(\alpha\)}

The parameter \(\alpha\) governs the chance‐constraint confidence \(1-\alpha\) in Phase~I, thus controlling risk tolerance.  To systematically evaluate its impact, we swept \(\alpha\) over the set \(\{0.01,\,0.05,\,0.10,\,0.15,\,0.20,\,0.25\}\) for two representative instances, \((n,T)=(30,20)\) and \((100,80)\).  Figure~\ref{fig_alpha} (left panels) shows how the damage objective \(Z_1\) decreases as \(\alpha\) increases—demonstrating the expected trade‐off from risk‐averse to risk‐seeking allocations.  The right panels plot the corresponding points in the \((Z_1,Z_2)\) space, with the subset of solutions satisfying Pareto‐optimality highlighted in red.

This explicit \(\alpha\)-parameter sweep confirms that:
\begin{itemize}
  \item Lower \(\alpha\) values (e.g.\ 0.01) produce conservative allocations with low damage but higher transfer costs.
  \item Higher \(\alpha\) values (e.g.\ 0.25) accept more damage in order to reduce transfer cost.
  \item Intermediate values trace a smooth trade‐off curve, enabling decision‐makers to select an allocation that matches their risk preference.
\end{itemize}

\begin{figure}[!t]
    \centering
    \begin{subfigure}{0.48\linewidth}
        \centering
        \includegraphics[width=\linewidth]{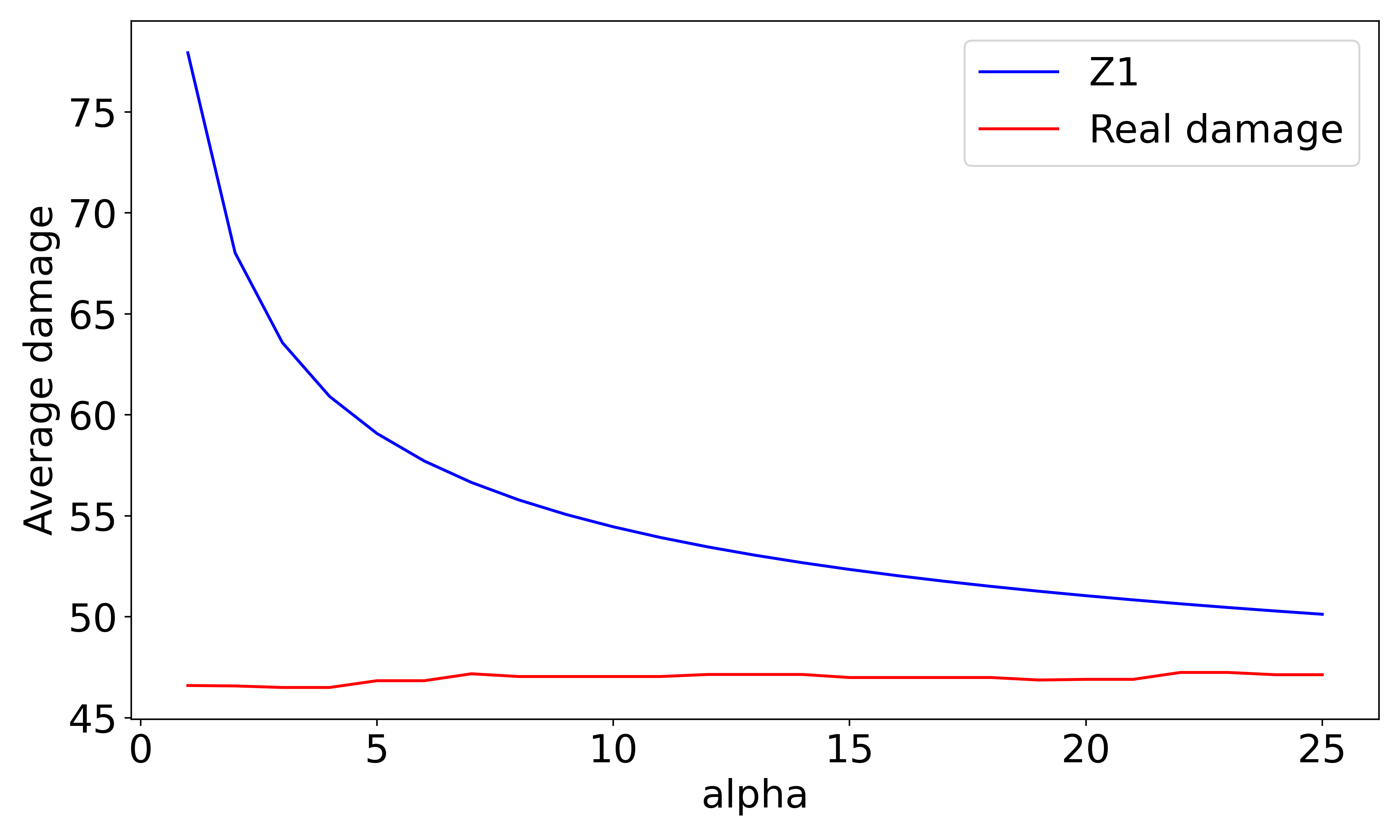}
        \caption*{$Z_1$ values and real damage values for instance ($30, 20$)}
    \end{subfigure}
    \hfill
    \begin{subfigure}{0.50\linewidth}
        \centering
        \includegraphics[width=\linewidth]{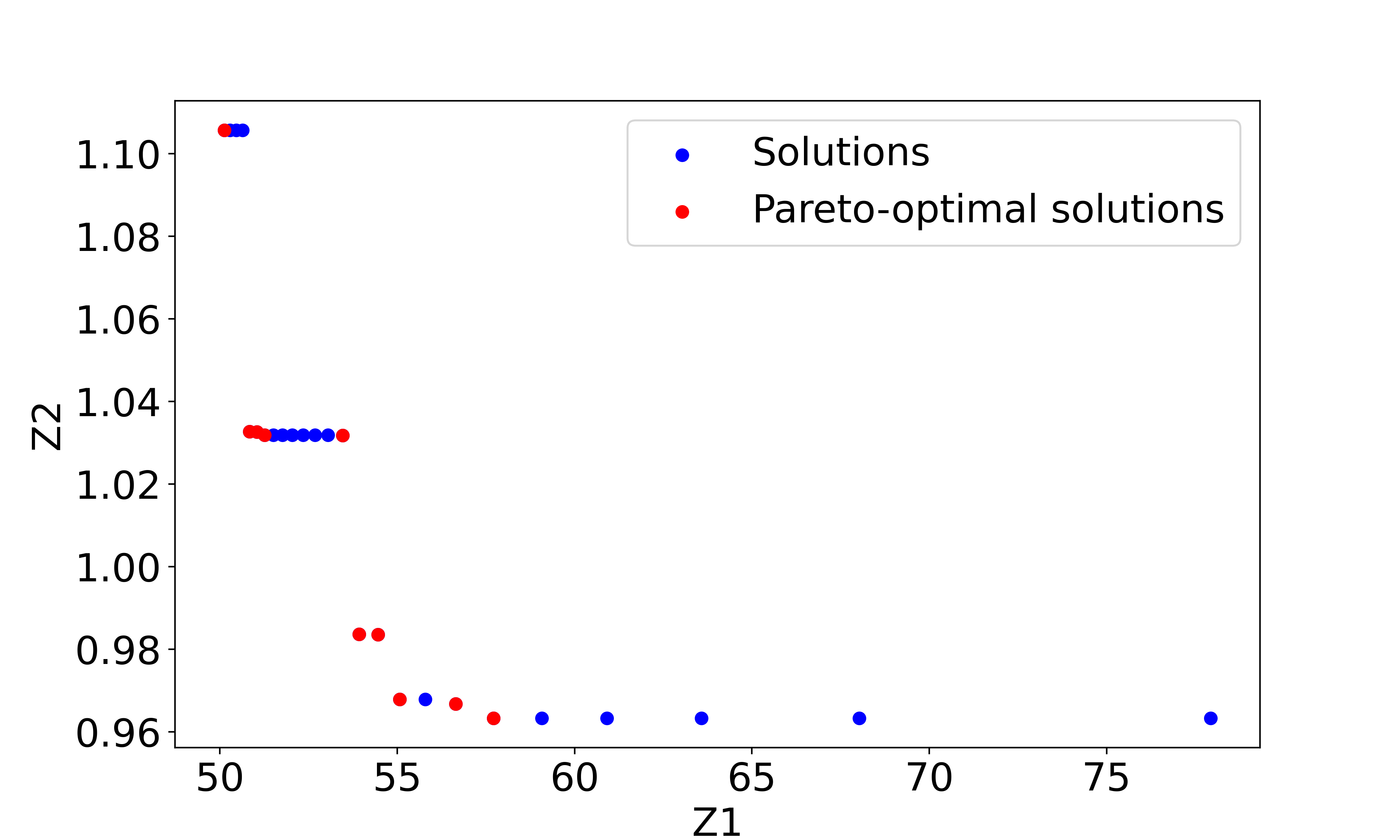}
        \caption*{Obtained solutions and Pareto-optimal solutions for ($30, 20$)}
    \end{subfigure}

    \begin{subfigure}{0.48\linewidth}
        \centering
        \includegraphics[width=\linewidth]{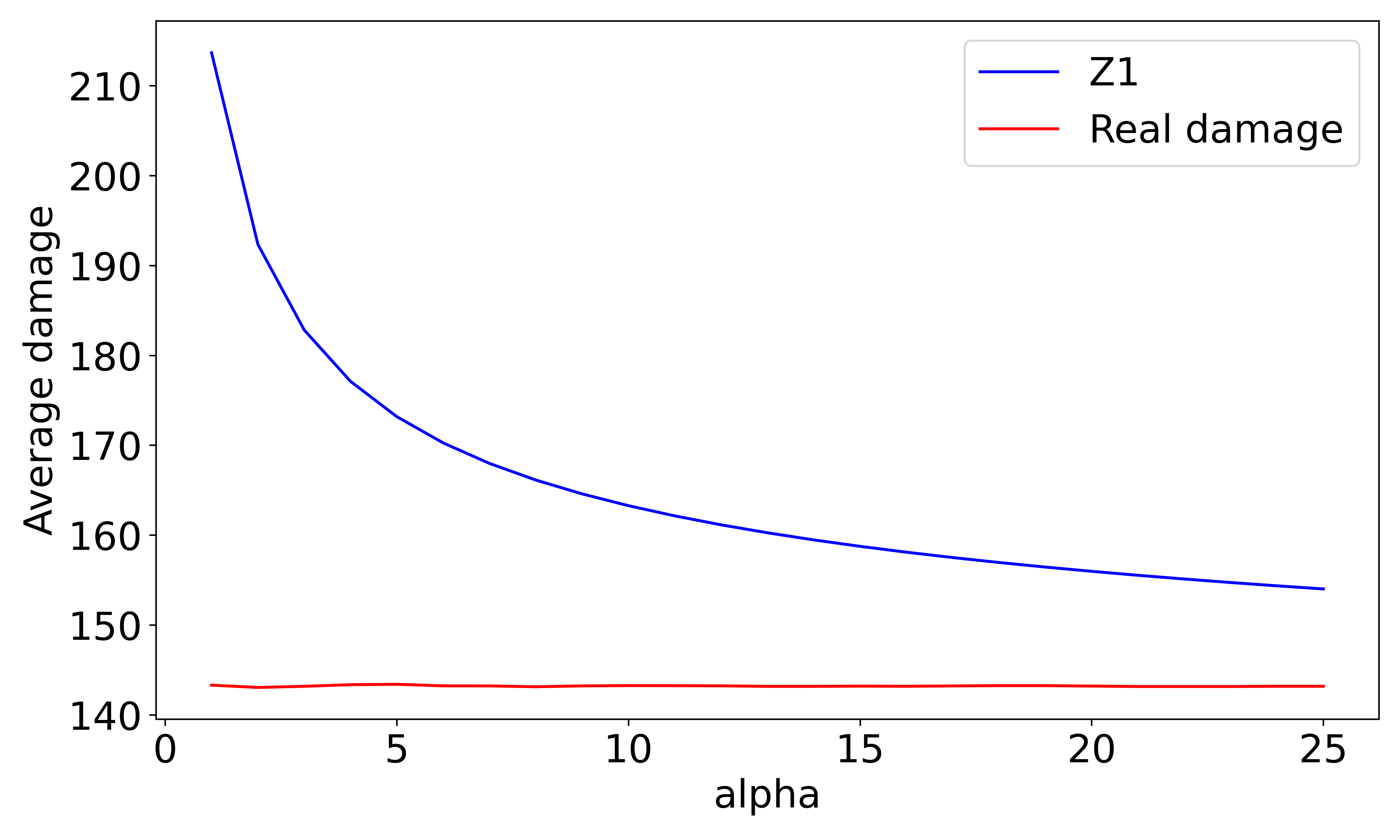}
        \caption*{$Z_1$ values and real damage values for instance ($100, 80$)}
    \end{subfigure}
    \hfill
    \begin{subfigure}{0.5\linewidth}
        \centering
        \includegraphics[width=\linewidth]{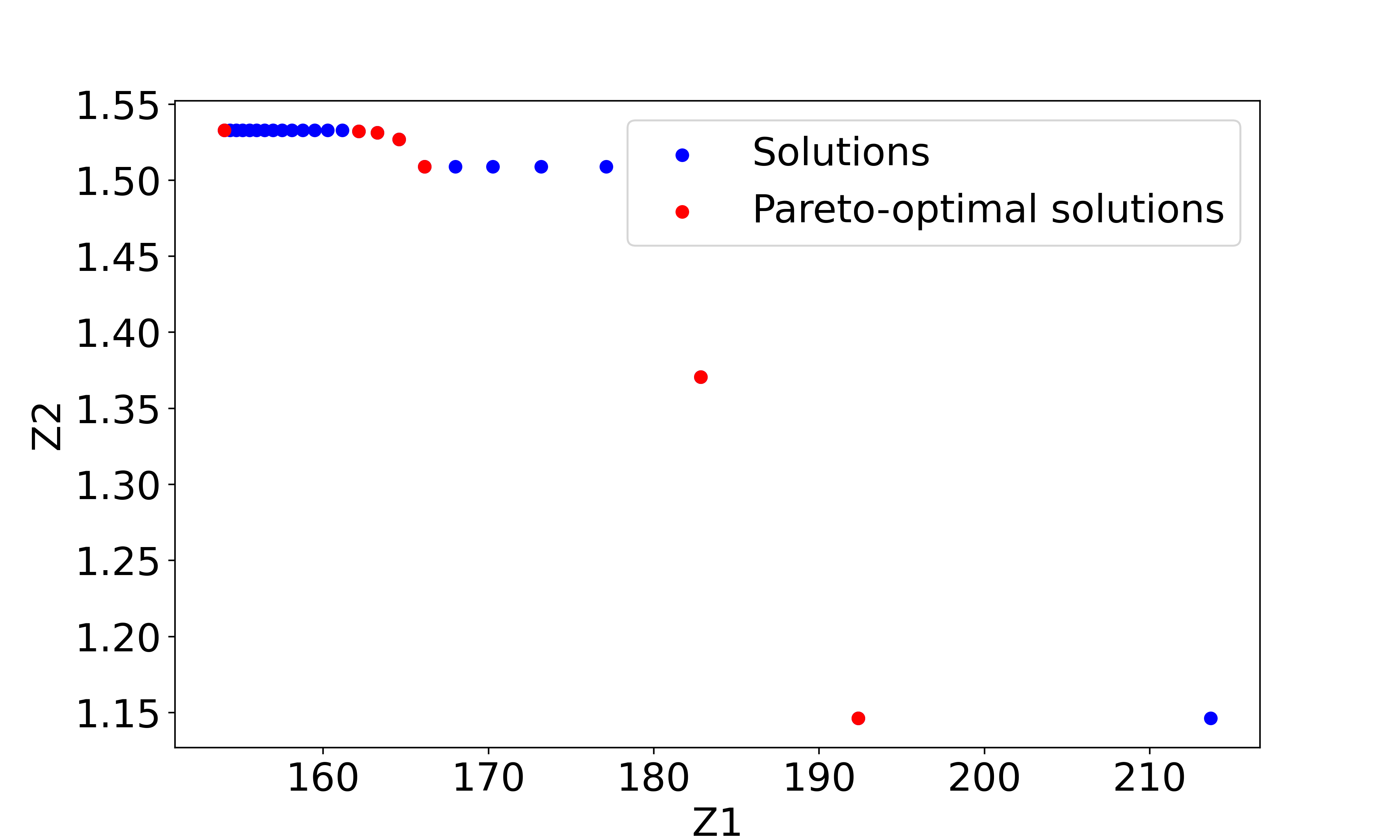}
        \caption*{Obtained solutions and Pareto-optimal solutions for ($100, 80$)}
    \end{subfigure}

    \caption{Solutions obtained for various values of the parameter $\alpha$ between 0.01 and 0.25. The left panels display the damage values corresponding to each $\alpha$, and the right panels illustrate the associated solutions in the objective space.}
    \label{fig_alpha}
\end{figure}
\subsection{Effect of Attack Probability}
To assess the impact of the attack probability on the system damage, we select it randomly in ($0,p$), where $p$ increases from $0$ to $1$. In each run, we compute the damage using our algorithm and the Oracle. Figure \ref{fig_attack} depicts the results for two instances ($n=20, T=20$) and ($n=100, T=80$).  Both diagrams show that the damage value rises with the attack probability. However, there is a notable difference between the two methods: For low attack probabilities, the Oracle achieves significantly lower damage values as it has a perfect knowledge of future attacks. For higher attack probabilities, the gap between the results obtained by the Un-mean algorithm and the Oracle diminishes. As the probability of attacks increases, more nodes are targeted in each time slot. Consequently, the Oracle must allocate resources across a broader set of nodes to mitigate potential damage, leading to damage values that are closer to those obtained by the Un-mean algorithm. 
\begin{figure}[!t]
    \centering
    \begin{subfigure}{0.48\linewidth}
        \centering
        \includegraphics[width=\linewidth]{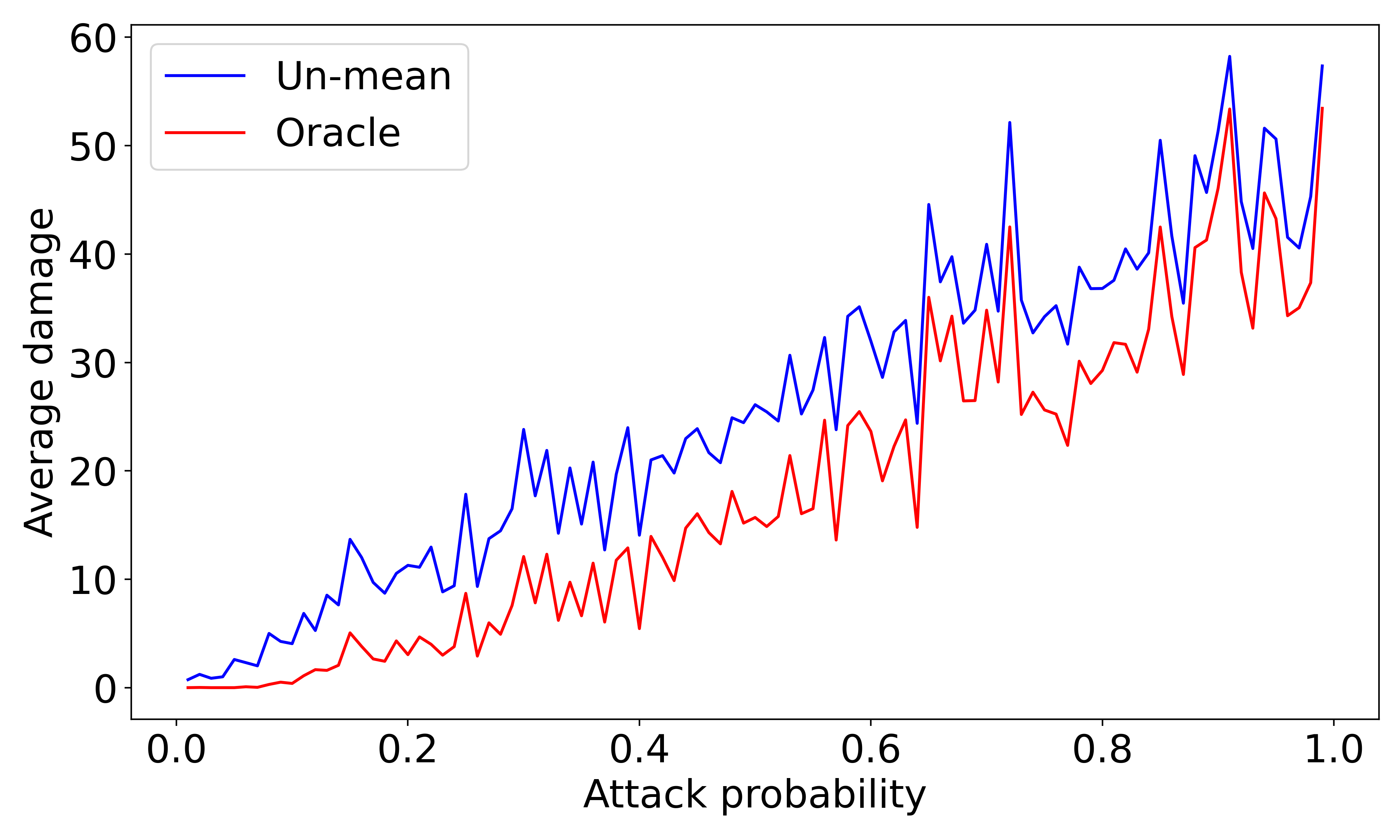}
        \caption*{Damage values for ($ 30,  20$)}
    \end{subfigure}
    \hfill
    \begin{subfigure}{0.48\linewidth}
        \centering
        \includegraphics[width=\linewidth]{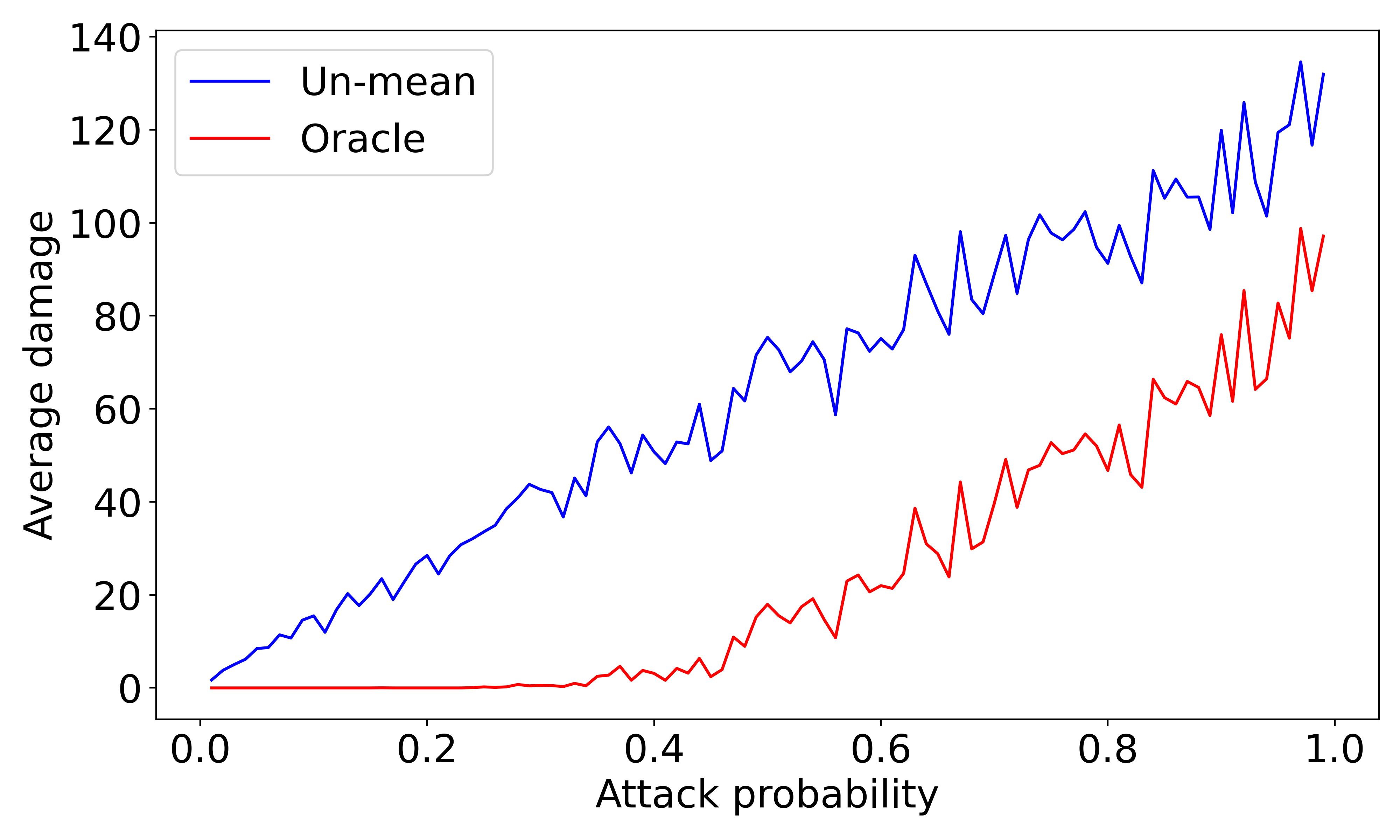}
        \caption*{Damage values for ($100, 80$)}
    \end{subfigure}
    \caption{Damage values obtained for various probabilities of attacks by Un-mean algorithm and Oracle under different attack scenarios.}
    \label{fig_attack}
\end{figure}
\subsection{Effect of Resource Availability}
We investigate the impact of the total resource availability, $R$, on the damage values. Within the constraints, $R$ ranged from $\sigma_i^n r_i^{\text{min}}$ to $\sigma_i^n r_i^{\text{max}}$. Values below this range result in no feasible solution, whereas those above result in perfect resource allocation and zero damage. We ran both the Un-mean algorithm and the Oracle for various resource values within this range and computed the corresponding damage values. Figure \ref{fig_resource} show the results for two different instance ($n=20, T=20$) and ($n=100, T=80$).

The damage values for both of the algorithms decrease with more resources. However, the two methods perform similarly for resource values close to the lower and upper bounds of the range. With scarce resources, both algorithms prioritize satisfying the minimum resource constraints for all nodes before distributing any remaining resources. That makes them vulnerable to attacks, leading to comparable damage values. With abundant resources, both algorithms fulfill the maximum required resources for all nodes, rendering them impervious to damage. For resource values in the middle range, particularly around the midpoint, the algorithms show the maximum disparity because Oracle has perfect knowledge about the nodes targeted in the next time slot. 
\begin{figure}[!t]
    \centering
    \begin{subfigure}{0.48\linewidth}
        \centering
        \includegraphics[width=\linewidth]{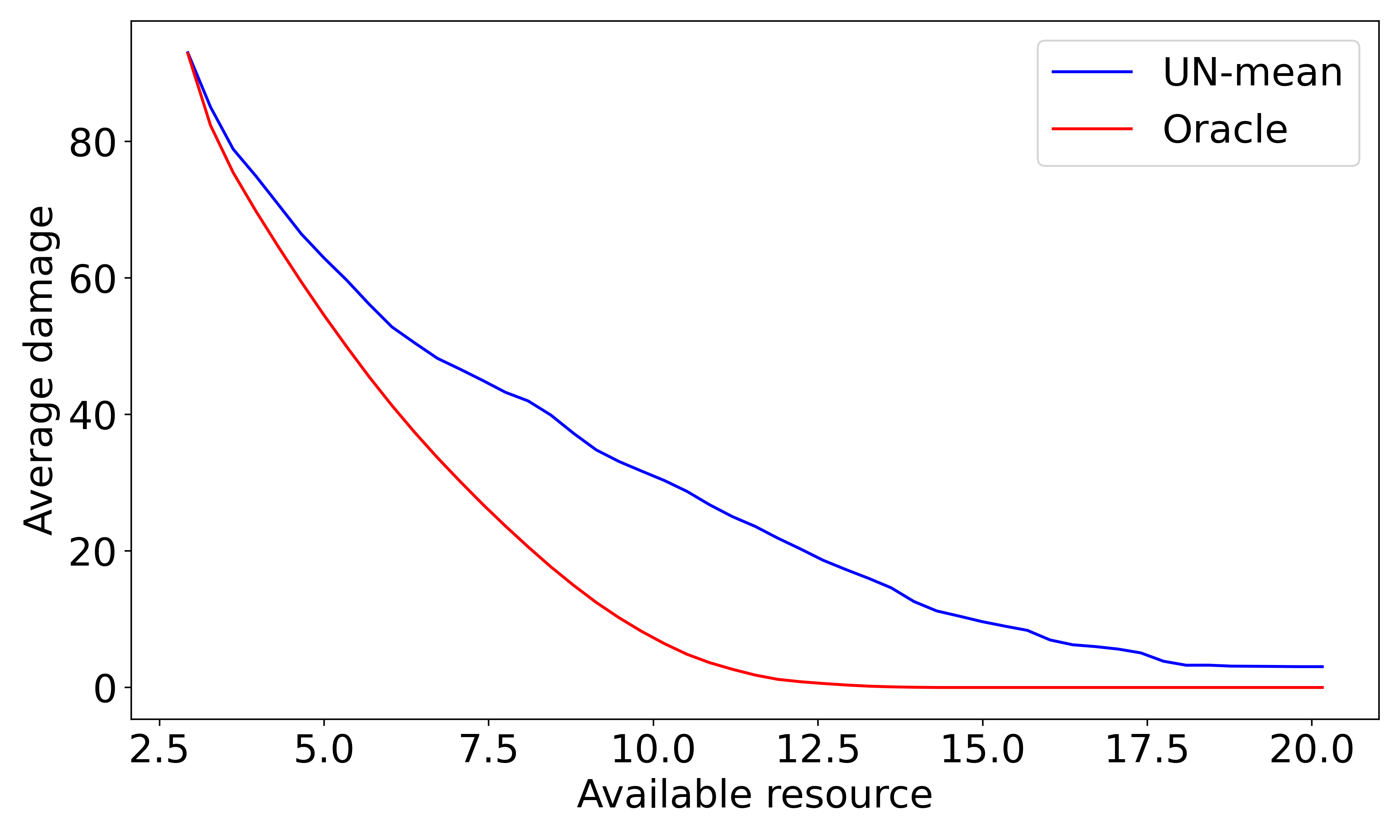}
        \caption*{Damage values for ($30,20$)}
    \end{subfigure}
    \hfill
    \begin{subfigure}{0.48\linewidth}
        \centering
        \includegraphics[width=\linewidth]{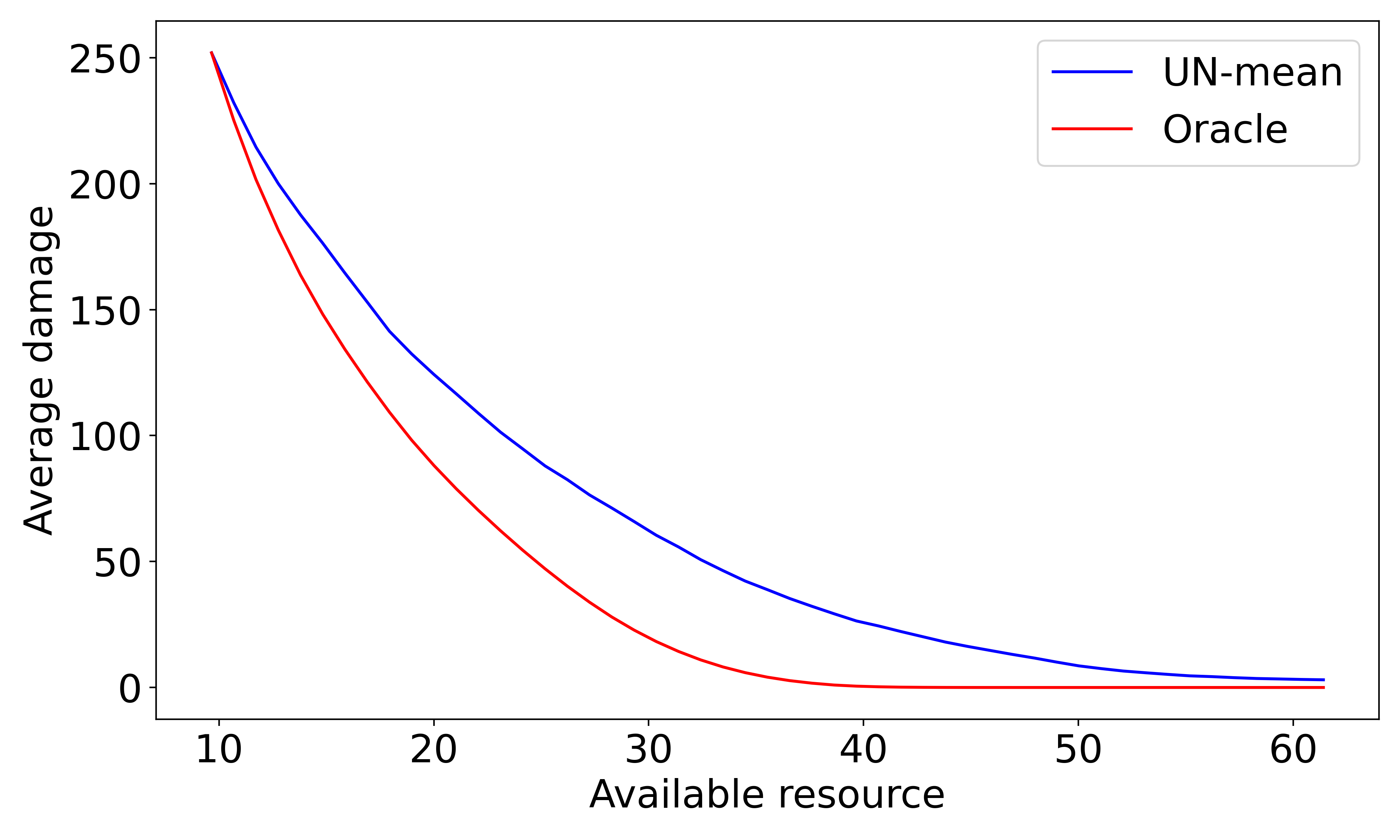}
        \caption*{Damage values for ($100, 80$)}
    \end{subfigure}

    \caption{Damage values obtained for various probabilities of attacks by Un-mean algorithm and Oracle under different resource availability.}
    \label{fig_resource}
\end{figure}
\section{Conclusion and Future work}
In this paper, we minimized the a priori unknown adversary's impact on the overall performance of a centralized system by optimal resource allocation while maintaining a low resource transfer cost. We formulated the problem within a bi-objective optimization framework and decomposed it into two components: Optimizing resource allocation and minimizing resource transfer among nodes. We presented quadratic chance-constrained and linear network flow programming models for the first and second components. Our methodology offered significant advantages in effective risk management, balancing the risk and performance objectives based on the decision-maker preferences, and enhancing model robustness against uncertainty. Future works include enabling distributed decision-making at the node level and including a strategic adversary that adapts to resource allocation decisions. 
\bibliographystyle{IEEEtran}

\newpage
\section*{Appendix}
In this section, we provide a detailed description of each component of the proposed algorithm.

\subsection*{Chance-constrained Approach for Subproblem I}
We reformulate the first subproblem as a chance-constrained program, where we impose probabilistic constraints to manage the risk associated with uncertain outcomes. It improves the probability of meeting the constraints under uncertainty, making the solution robust in a probabilistic sense.
\begin{align}
\label{eq:CCP}
\text{Minimize} \quad & Z_1^t = \epsilon^t \nonumber \\
\text{Subject to:} \quad 
& \text{Pr}\left( \sum_{i=1}^{n} w_i \tilde{y}_i^t \frac{r_i^{\text{max}} - r_i^t}{r_i^{\text{max}} - r_i^{\text{min}}} \leq \epsilon^t \right) \geq 1 - \alpha, && \forall t \nonumber \\
& \sum_{i=1}^{n} r_i^t \leq R, && \forall t \nonumber \\
& r_i^{\text{min}} \leq r_i^t \leq r_i^{\text{max}}, && \forall i, \forall t \nonumber
\end{align}

where $\epsilon^t$ is an upper bound for the total damage in each time slot $t$, and $\alpha$ is the risk parameter. For example, $\alpha=0.05$ ensures that with at least 95\% confidence the system's damage does not exceed $\epsilon^t$, which is indeed the minimum possible damage for such a confidence level. In other words, incorporating the first constraint of \eqref{eq:CCP} enables us to quantify and manage the uncertainty in the system while optimizing for desired objectives. 

To solve \eqref{eq:CCP}, we write its probabilistic constraint, i.e., $\text{Pr}\left(\sum_{i=1}^{n} w_i \tilde{y}_i^t \frac{r_i^{\text{max}} - r_i^t}{r_i^{\text{max}} - r_i^{\text{min}}} \leq \epsilon^t\right) \geq 1-\alpha$, as the following deterministic constraint \cite{nikolova2010approximation,yang2017algorithm}. 
\begin{align}
\sum_{i=1}^{n} \mathbb{E}[\tilde{y}_i^t] v_i^t + \sqrt{\frac{1-\alpha}{\alpha}} \sqrt{ \sum_{i=1}^{n} \mathbb{V}[\tilde{y}_i^t] v_i^t} \leq \epsilon^t,
\end{align}
where $\mathbb{E}[\tilde{y}_i^t]$ and $\mathbb{V}[y_i^t]$ are the mean and variance of $\tilde{y}_i^t$, and $v_i^t = w_i \frac{r_i^{\text{max}} - r_i^t}{r_i^{\text{max}} - r_i^{\text{min}}}$ is a variable substitution for convenience. With a simple variable substitution of $h^2 = \sum_{i=1}^{n} \mathbb{V}[\tilde{y}_i^t] v_i^t$. So, we can write
\begin{align}
& \sum_{i=1}^{n} \mathbb{E}[\tilde{y}_i^t] v_i^t + \sqrt{\frac{1-\alpha}{\alpha}} h  \leq \epsilon,
\label{EQchance_n}
\end{align}
which renders Subproblem 1 a \textit{quadratic programming}.

As mentioned earlier, the adversary's attack distribution function and the exact values for $\mathbb{E}[\tilde{y}_i^t]$ and $\mathbb{V}[y_i^t]$ are not available. Consequently, we develop an estimation method that involves iterative solutions to Subproblem 1 over $t=1,2,\ldots,T$, followed by updates based on the observed attacks at each time slot.

\subsection*{Network Flow Approach for Subproblem II}
Knowing the allocated resource to each node $i$ at time $t-1$, i.e., $r_i^{t-1}$, Subproblem 2 optimizes the redistribution of resources among network nodes to ensure the optimal allocation for each node $i$ at time $t$, i.e., $(r_i^t)^*$, which is calculated by Subproblem 1. To this end, we convert subproblem 2 to a weighted network flow problem and solve it efficiently.

The weighted network flow problem involves determining the optimal flow of resources through a network while minimizing the total cost. To convert Subproblem 2 to a weighted network flow problem, we insert two artificial \textit{source} and \textit{sink} nodes, denoted by index $i=0$ and $i=n+1$, respectively. We set the cost for the edges such artificial edges equal to 0, i.e., $c_{0i}$ and $c_{i(n+1)}$ to $0$ for all $i\in \{1,2,...,n\}$. Also, we consider a sufficiently large capacity, e.g., $u_{ij}= R$, for a directed edge originating from node $i$ to node $j$. Additionally, we introduce an edge from the source node to a node $i$ if $(r_i^t)^* < (r_i^{t-1})^*$, with $u_{0i} = (r_i^{t-1})^* - (r_i^t)^*$. We use $A^t$ to represent these edges, i.e., $A^t=\{i\in \{1,2,...,n\}|(r_i^{t-1})^*>(r_i^t)^* \}$. Similarly, we establish an edge from a node $i$ to the sink node if ($r_i^{t-1})^* < (r_i^t)^*$ and assign $u_{i(n+1)} = (r_i^t)^* - (r_i^{t-1})^*$. We use $B^t$ to represent these edges, i.e., $B^t=\{i\in \{1,2,...,n\}|(r_i^{t-1})^* <  (r_i^t)^* \}$. Under the following initialization of the input parameters, we can formulate the problem as a linear program:\\
\textbf{Initialization of the added parameters:}
\begin{align*}
    & c_{0i} = 0, \quad \text{if} \quad i \in A^t \\
    & c_{i(n+1)} = 0, \quad \text{if} \quad i \in B^t\\
    & u_{ij} = R, \quad \forall i,j \in \{1,2,\ldots,n\}\\
    & u_{0i} = (r_i^{t-1})^* - (r_i^t)^*, \quad \text{if} \quad i \in A^t\\
    & u_{i(n+1)} = (r_i^t)^*-(r_i^{t-1})^*, \quad \text{if} \quad i \in B^t
\end{align*}
\textbf{Linear programming model for the resource transfer:}
\begin{align*}
    \text{Minimize} \quad & \sum_{i,j} c_{ij}  x_{ij}^t \\
    \text{Subject to:}\\
    & \sum_{i} x_{ij}^t - \sum_{i} x_{ji}^t = 0 , \quad \forall j \in \{1,2,\ldots,n\}\\
    & \sum_{i} x_{i(n+1)}^t = \sum_{i \in B^t} (r_i^t)^* - (r_i^{t-1})^*   \\
    & 0 \leq x_{ij}^t \leq u_{ij}, \quad \forall i,j \in \{0,1,\ldots,n+1\}
\end{align*}
The objective of the problem is to minimize the total cost of resource transfer. The first constraint ensures that the transferred resources into a node equals the resource out of that node (except for the source and sink nodes). The second constraint guarantees that the resource flow is sufficient for all nodes in $(B^t)$ to achieve their optimal resource levels in time slot $t$. The third constraint stipulates that the resources on each edge do not exceed its capacity, thereby guaranteeing the minimum requirement.
\subsection*{Algorithm for Solving RAPA}
The algorithm starts with the initial values $\mathbb{E}[\tilde{y}_i^1] = \frac{1}{n}$ and $\mathbb{V}[y_i^1] = 0$ for $i=1,2,\dots, n$. Also, we initialize resource allocation by the feasible solution $r_i^0 = r_i^{\text{min}} + \frac{w_i}{\sum_{j=0}^n w_j} (R-\sum_{j=1}^n r_j^{\text{min}})$. That implies that in the absence of any historical information about the attacks, we allocate resources based on the importance weights. In each round $t \geq 1$, we compute the optimal resource allocations $(r_i^{t})^*$ and then the optimal resource transfer $(x_{ij}^1)^{t}$ for $i,j \in \{1,2,\dots,n \}$. We apply the computed allocations, observe the attacks, and update the mean and variance of attacks, e.g., using a Bayesian strategy.
\begin{algorithm}[H]
\caption{Algorithm for RAPA}
\label{alg:appendix}
\begin{algorithmic}[1]
\STATE \textbf{Initialize} $\mathbb{E}[\tilde{y}_i^1]=\frac{1}{n}$, $\mathbb{V}[y_i^1]=0$ and $r_i^0 = r_i^{\text{min}} + \frac{w_i}{\sum_{j=1}^n w_j} (R-\sum_{j=1}^n r_j^{\text{min}})$ for $i=1,2,\dots,n$.
\FOR{$t=1$ to $T$}
    \STATE \textbf{Find} optimal resource allocation $(r_i^t)^*$ by solving:
    \STATE \quad \fbox{\parbox{0.92\linewidth}{%
    \begin{align*}
    & \text{Minimize} && Z_1^t = \epsilon^t \\
    & \text{Subject to:}  \\
    &&& \sum_{i=1}^{n} \mathbb{E}[\tilde{y}_i^t] v_i^t + \sqrt{\frac{1-\alpha}{\alpha}} h \leq \epsilon^t,  \\
    &&& h^2 = \sum_{i=1}^{n} \mathbb{V}[\tilde{y}_i^t] v_i^t, \\
    &&&  v_i^t = w_i \frac{r_i^{\text{max}} - r_i^t}{r_i^{\text{max}} - r_i^{\text{min}}}, \quad \forall i \\
    &&& \sum_{i=1}^{n} r_i^t \leq R, \\ 
    &&& r_i^{\text{min}} \leq r_i^t \leq r_i^{\text{max}}, \quad \forall i 
    \end{align*}
    }}
    \STATE \textbf{Find} optimal resource transfer $(x_{ij}^t)^*$ by solving:
    \STATE \quad \fbox{\parbox{0.92\linewidth}{%
    \begin{align*}
    \text{Minimize} &\quad Z_2^t = \sum_{i,j} c_{ij}  x_{ij}^t \\
    \text{Subject to:}\\
    & \sum_{i} x_{ij}^t - \sum_{i} x_{ji}^t = 0 , \quad \forall i \\
    & \sum_{i} x_{i(n+1)}^t = \sum_{i \in B^t} (r_i^t)^* - (r_i^{t-1})^*, \forall i   \\
    & 0 \leq x_{ij}^t \leq u_{ij},  \forall i,j \in \{0,\ldots,n+1\}
    \end{align*}
    }}
    \STATE \textbf{Apply} the obtained resource allocation, observe the actual attack scenario,  and update the mean and variance of attacks for all nodes $i=1,2,\dots, n$:
    \begin{align*}
    & \mathbb{E}[\tilde{y}_i^{t+1}]=\frac{t}{t+1} \mathbb{E}[\tilde{y}_i^{t}] + \frac{\text{Observed} ({y}_i^{t})} {t+1} \\
    & \mathbb{V}[\tilde{y}_i^{t+1}] = \frac{t}{t+1} \mathbb{V}[\tilde{y}_i^t] + \frac{(\text{Observed} ({y}_i^{t}) - \mathbb{E}[\tilde{y}_i^{t+1}])^2} {t+1}
    \end{align*}
\ENDFOR
\STATE \textbf{return} $(r_i^t)^*$ and $(x_{ij}^t)^*$ for all $t \in \{1,2,\ldots,T \}$ and $i,j \in \{1,2,\ldots,n\}$.
\end{algorithmic}
\end{algorithm}

\end{document}